\documentclass[twoside,10pt]{article}

\usepackage{microtype}
\usepackage{graphicx}
\usepackage{subcaption}
\usepackage{booktabs} % for professional tables

\usepackage{hyperref}
\newcommand{\newcheckmark}{\raisebox{0.6ex}{\scalebox{0.7}{$\sqrt{}$}}}
\newcommand{\newcrossmark}{\scalebox{0.85}[1]{$\times$}}

\usepackage[abbrvbib, preprint]{jmlr2e}

\usepackage{amsmath,amsfonts,bm}

\def\1{\bm{1}}

\DeclareMathAlphabet{\mathsfit}{\encodingdefault}{\sfdefault}{m}{sl}
\SetMathAlphabet{\mathsfit}{bold}{\encodingdefault}{\sfdefault}{bx}{n}

\DeclareMathOperator*{\argmax}{arg\,max}

\usepackage{url}

\usepackage{physics}
\usepackage{multirow}
\usepackage[ruled,vlined]{algorithm2e}
\SetKwInOut{Parameter}{Parameters}
\usepackage{tikz}
\usepackage{tikz-cd}

\allowdisplaybreaks

\def\tilde{\widetilde}
\def\t{\tilde}
\def\hat{\widehat}

\allowdisplaybreaks

\usepackage{amsmath}
\usepackage{amssymb}
\usepackage{mathtools}
\usepackage[capitalize,noabbrev]{cleveref}

\ShortHeadings{Learning to Coordinate via Quantum Entanglement in Multi-Agent Reinforcement Learning}{Learning to Coordinate via Quantum Entanglement in Multi-Agent Reinforcement Learning}
\firstpageno{1}

\usepackage{lastpage}
\jmlrheading{23}{2026}{1-\pageref{LastPage}}{MM/YY; Revised MM/YY}{9/22}{21-0000}{Nasdaq, Inc.}

\begin{document}

\title{Learning to Coordinate via Quantum Entanglement\\ in Multi-Agent Reinforcement Learning}

% % Option 1
% \author{%
% \begin{center}
%        \name John Gardiner$^{*,1}$
%        \name Orlando Romero$^{*,1}$
%        \name Brendan Tivnan$^{1}$ \\
%        \name Nicol\`o Dal Fabbro$^{1,2}$
%        \name George J. Pappas$^{2}$
% \end{center}
% }

% Option 2
\author{%
       \name John Gardiner$^*$ \email john.gardiner@nasdaq.com \\
       \addr Nasdaq, Inc.
       \AND
       \name Orlando Romero$^*$ \email orlando.romero@nasdaq.com \\
       \addr Nasdaq, Inc.
       \AND
       \name Brendan Tivnan \email brendan.tivnan@nasdaq.com \\
       \addr Nasdaq, Inc.
       \AND
       \name Nicol\`o Dal Fabbro \email nicolo.dalfabbro@nasdaq.com \email \\
       \addr Nasdaq, Inc. \\
       \addr University of Pennsylvania
       \AND
       \name George J. Pappas \email pappasg@seas.upenn.edu \\
       \addr University of Pennsylvania
}

% % Option 3
% \author{%
%        \name John Gardiner$^*$ \email \addr Nasdaq, Inc.
%        \AND
%        \name Orlando Romero$^*$ \email \addr Nasdaq, Inc.
%        \AND
%        \name Brendan Tivnan \email \addr Nasdaq, Inc.
%        \AND
%        \name Nicol\`o Dal Fabbro \email \addr Nasdaq Inc., UPenn
%        \AND
%        \name George J. Pappas \email \addr \addr UPenn
% }

\maketitle

% \vspace{-3.5cm}
\vspace{-0.75cm}
% \vspace{-1cm}

\begin{center}
    \footnotesize $^*$Equal contribution
\end{center}

% \begin{center}
%     \footnotesize $^*$Equal contribution. Correspondence: john.gardiner@nasdaq.com
% \end{center}

% \begin{center}
%     \footnotesize $^*$Equal contribution. $^1$Nasdaq, Inc. $^2$University of Pennsylvania. \\ Correspondence: john.gardiner@nasdaq.com, orlando.romero@nasdaq.com
% \end{center}

\vspace{0.5cm}
% \vspace{0.05cm}
% \vspace{-1cm}

\begin{abstract}
The inability to communicate poses a major challenge to coordination in multi-agent reinforcement learning (MARL). 
Prior work has explored correlating local policies via \emph{shared randomness}, sometimes in the form of a \emph{correlation device}, as a mechanism to assist in decentralized decision-making. 
In contrast, this work introduces the first framework for training MARL agents to exploit shared \mbox{\emph{quantum entanglement}} as a coordination resource, which permits a larger class of communication-free correlated policies than shared randomness alone. 
This is motivated by well-known results in quantum physics which posit that, for certain single-round cooperative games with no communication, shared quantum entanglement enables strategies that outperform those that only use shared randomness. In such cases, we say that there is \mbox{\textit{quantum advantage}}.
Our framework is based on a novel differentiable policy parameterization that enables optimization over quantum measurements, together with a novel policy architecture that decomposes joint policies into a quantum coordinator and decentralized local actors.
{To illustrate the effectiveness of our proposed method, we first show that we can learn, purely from experience, strategies that attain quantum advantage in single-round games that are treated as black box oracles. 
We then demonstrate how our machinery can learn policies with quantum advantage in an illustrative multi-agent sequential decision-making problem formulated as a decentralized partially observable Markov decision process (Dec-POMDP).}
\end{abstract}

%%%%%%%%%%%%%%%%%%%%%%%%%%%%%%%%%%%%%%%%%%%%%%%%%%%%%%%%%%%%%%
%%%%%%%%%%%%%%%%%%%%%%%% Introduction %%%%%%%%%%%%%%%%%%%%%%%%
%%%%%%%%%%%%%%%%%%%%%%%%%%%%%%%%%%%%%%%%%%%%%%%%%%%%%%%%%%%%%%
\section{Introduction}
A key challenge in multi-agent reinforcement learning (MARL) is partial observability (sometimes called \mbox{\emph{information asymmetry}}), where each
agent possesses different information about the state and action processes while making decisions~\cite{kim2020communication, kao2022common, liu2023partially}.

\noindent\textit{Communication constraints.}  
Inter-agent communication is a common mechanism to combat information asymmetry~\citep{foerster2016learning, liu2023partially, schroeder2019multi}. However, in many scenarios of practical interest, such as in low-latency decentralized decision making (e.g., high-frequency trading~\cite{ding2024coordinating}) or military applications, communication is too costly, impossible, or otherwise constrained.

\noindent\textit{Shared randomness.}
Correlating policies in MARL via shared randomness has primarily been explored in the form of a \emph{correlation device} (see Section 6.2.4 of~\citet{Amato2016}'s tutorial and references therein). However this research largely focuses on computational efficiency, rather than expressiveness of policy classes.

% \pagebreak

\noindent\textit{Quantum entanglement.} %Absence of communication can severely limit the agents’ ability to coordinate their behaviors.
The seminal work by \citet{bell1964einstein} proved that quantum entanglement can result in decentralized, yet correlated, decisions that cannot be reproduced classically (that is, in the absence of entanglement). Subsequent works~\cite{original_chsh, 1313847} demonstrated the benefit that quantum entanglement can provide for cooperative decision making without communication.
See \cite{RevModPhys.86.419} for a standard review.
Despite continued research on single-round nonlocal games \cite{vaidman1999variations, mironowicz2023entangled, PhysRevA.109.042201}, analysis of the potential benefit of quantum entanglement for sequential decision-making problems is rather recent \cite{da2025entanglement, dasilva2026entanglement}. 

In this work, we contribute to the above line of research, asking the following question:

\begin{center}
    \textit{Can we learn communication-free policies that exploit\\ quantum entanglement for coordination in MARL?}
\end{center}

\textbf{Contributions.} 
We respond in the affirmative. Our contributions can be summarized as follows: 
\begin{itemize} 
    \item We delineate a hierarchy of joint policy classes for communication-free cooperative MARL (see Figure~\ref{fig:hierarchy-of-policies}), which notably includes \emph{shared randomness}\footnote{Which includes \emph{factorized} policies $\pi(\mathbf{a}|\mathbf{h}) = \prod_i\pi_i(a_i|h_i)$.} policies, \emph{shared quantum entanglement} policies (a superset to shared randomness policies, still implementable in a decentralized manner), and \emph{non-signaling} policies (the most general form of communication-free policies).
    \item We introduce the first MARL framework in which decentralized agents learn to exploit \emph{quantum entanglement} as a shared coordination resource. First, we develop $\mathsf{QuantumSoftmax}$ (see Algorithm~\ref{alg:quantum-softmax}), a differentiable transformation that maps arbitrary square complex-valued matrices to an object that formally describes a quantum measurement. Thus, our framework enables end-to-end gradient-based optimization over quantum measurements.  Second, we introduce an \emph{advice}-based policy architecture that cleanly separates joint policies into a \emph{quantum coordinator}, which samples correlated advice via quantum measurements, and local actors that condition on this advice, allowing seamless integration with policy gradient methods. Third, we instantiate this framework in a modified multi-agent proximal policy optimization (MAPPO) algorithm capable of learning entangled policies for sequential decision-making. 
    \item We first validate our framework on single-round cooperative games with theoretically established quantum advantage, confirming that our algorithm recovers known optimal entangled strategies. We then apply our framework to a multi-router multi-server queueing problem formulated as a decentralized partially observable Markov decision process (Dec-POMDP), where we learn sequential decision-making policies that achieve quantum advantage in a setting previously analyzed only through queueing-theoretic methods. 
\end{itemize}

%%%%%%%%%%%%%%%%%%%%%%%%%%%%%%%%%%%%%%%%%%%%%%%%%%%%%%%%%%%%%%
%%%%%%%%%%%%%%%%%%%%%%%% Related Work %%%%%%%%%%%%%%%%%%%%%%%%
%%%%%%%%%%%%%%%%%%%%%%%%%%%%%%%%%%%%%%%%%%%%%%%%%%%%%%%%%%%%%%

\textbf{Related Work.}
The use of quantum entanglement for distributed decision-making has a long history, often focused on the study of single-round cooperative games called nonlocal games \cite{1313847}. Though many such games with quantum advantage have been studied, they are often quite stylized. Some recent works attempt to show quantum advantage in single-round decision-making tasks with a more applied flavor, for example in rendezvous problems over graphs \cite{PhysRevA.109.042201}, in a simplified high frequency trading scenario \cite{ding2024coordinating}, and in a load balancing problem in ad hoc networks \cite{Hasanpour:2017uqy}.

While optimal quantum entangled strategies for these single-round games are often found analytically or known by construction, there is work using gradient-based approaches to discover such optimal strategies, for example \cite{bharti2019teach} and~\cite{furches2025application}. The idea of learning strategies from experience, with games treated as black boxes, was considered recently by~\citet{kerenidis2025quantum}.

The recent work of~\citet{da2025entanglement, dasilva2026entanglement}, a first of its kind, investigates the use of entanglement in decentralized, non-trivially sequential decision-making in a multi-router multi-server queueing problem whose special form allows one to theoretically show quantum advantage.
They do not use machine learning methods to discover strategies.

In contrast to the prior research, our work is the first to investigate how to learn parameterized MARL policies that exploit entanglement {during execution} to obtain a coordination advantage.

We also mention here a separate line of work which builds upon progress in quantum machine learning~\cite{oh2020tutorial}, extending it to quantum RL and quantum MARL~\cite{yun2022quantum, kwak2021introduction}. These efforts are fundamentally different from ours: they use quantum computing (typically parametrized quantum circuits) with the aim to reduce the computational burden of learning agents by reducing parameter count, required data, or training time. In this vein, the recent work by~\citet{derieuxeqmarl} studies how to leverage quantum entanglement during training in MARL to replace a centralized critic with a partially decentralized one. 
Notably, entanglement in their work serves a computational complexity purpose, to accelerate training and reduce parameter count relative to other solutions, not to achieve coordination advantages at execution time between non-communicating agents. Explicitly, they do not consider entanglement shared between actor models.

%%%%%%%%%%%%%%%%%%%%%%%%%%%%%%%%%%%%%%%%%%%%%%%%%%%%%%%%%%
%%%%%%%%%%%%%%%%%%%%%%%% Notation %%%%%%%%%%%%%%%%%%%%%%%%
%%%%%%%%%%%%%%%%%%%%%%%%%%%%%%%%%%%%%%%%%%%%%%%%%%%%%%%%%%
\textbf{Notation.}
% For a positive integer $n$, let $[n]\coloneqq\{1,\dots,n\}$. For a non-empty finite set $X$, we write $\Delta(X)$ for either the set of probability distributions over $X$ or the corresponding probability simplex. $\mathrm{expm}(A) := \sum_{k=0}^\infty\frac{A^k}{k!}$
$\Delta(X) \coloneqq  \{\text{probability distributions over } X\}$; $[n]\coloneqq\{1,\dots,n\}$; $\mathbf{x}_\mathcal{I} \coloneqq (x_i)_{i \in\mathcal{I}}$; $\mathrm{expm}(A) \coloneqq \sum_{k=0}^\infty\frac{A^k}{k!}$.

%%%%%%%%%%%%%%%%%%%%%%%%%%%%%%%%%%%%%%%%%%%%%%%%%%%%%%%%%%%%
%%%%%%%%%%%%%%%%%%%%%%%% Background %%%%%%%%%%%%%%%%%%%%%%%%
%%%%%%%%%%%%%%%%%%%%%%%%%%%%%%%%%%%%%%%%%%%%%%%%%%%%%%%%%%%%
\section{Relevant Quantum Theory}
\label{sec:quantum_theory}
This section describes the essential elements of quantum theory necessary for understanding what shared quantum entanglement entails for MARL in operational terms.
\subsection{Quantum systems and measurements}
Quantum theory stipulates how outcomes of measurements made on a quantum system relate to the physical state of that system.

\begin{definition}[Density matrix]
A \emph{density matrix} is a PSD matrix $\rho\in\mathbb{C}^{d\times d}$ such that $\mathrm{tr}(\rho) = 1$.
\end{definition}

It is postulated that every quantum system has a \emph{state}
\footnote{
    For those familiar with quantum physics: for simplicity, we will restrict ourselves to finite-dimensional systems throughout.
}, which can be represented as a density matrix $\rho$. A system or state is said to be $d$-dimensional when $\rho$ is of size $d\times d$. For a real-world quantum system, the density matrix might encode physical properties such as energy, spin, or momentum.

Note that convex combinations of density matrices are also density matrices.
The convex combination $\rho=\alpha\rho_1 + (1-\alpha)\rho_2$ represents the overall state of a system that is in state $\rho_1$ with probability $\alpha$ and state $\rho_2$ with probability $1-\alpha$.

A \emph{measurement} on a physical system is a procedure or experiment that can be performed on the system, resulting in some outcome dependent on the system state.
For example, in a physical system a measurement might measure the energy of the system, in which case the measurement outcome is the value of the energy.

\begin{definition}[POVM]
A \emph{positive operator-valued measure} (POVM) is a collection $M = (M_j)_{j\in [m]}$ of PSD matrices $M_1,\ldots,M_m \in \mathbb{C}^{d\times d}$ such that $M_1 + \ldots + M_m = I_d$.
\end{definition}
In quantum theory, a measurement with $m$ possible outcomes, compatible with a $d$-dimensional quantum system, is encoded as a POVM of the form $M \in (\mathbb{C}^{d\times d})^m$.
The outcome of such a measurement is probabilistic, with probabilities postulated by a law known as the 
\emph{Born Rule}:

\textbf{Born Rule}: Given a quantum system with state $\rho$, and given a measurement acting on the system, described by the POVM $M \in  (\mathbb{C}^{d\times d})^m$, we have $\mathbb{P}(\text{outcome $j$}) = \tr(\rho M_j)$ for each $j\in [m]$.
% \begin{equation}
% \mathbb{P}(\text{outcome $j$})
% =
% \tr(\rho M_j),
% \end{equation}

% \textbf{Born Rule:} Given a quantum system with state $\rho$, and given a measurement with associated POVM  $M \in  (\mathbb{C}^{d\times d})^m$ acting on that system, then
% \begin{equation}
% \mathbb{P}(\text{outcome $j$})
% =
% \tr(\rho M_j),
% \end{equation}
% for each $j \in [m]$.

\subsection{Joint systems and quantum entanglement}
\label{sec:joint_systems}
An important aspect of the quantum formalism is how it treats joint systems.
Given two quantum systems with dimension $d_1$ and dimension $d_2$ respectively, their \emph{joint system} has dimension $d_1d_2$.
Whereas the density matrices of the separate systems are elements of $\mathbb{C}^{d_1\times d_1}$ and $\mathbb{C}^{d_2\times d_2}$, density matrices of the joint system are trace 1, PSD elements of the much larger space
$\mathbb{C}^{d_1\times d_1}\otimes\mathbb{C}^{d_2\times d_2}\cong\mathbb{C}^{d_1d_2\times d_1d_2}$.

Given a state $\rho_A$ for system $A$ and a state $\rho_B$ for system $B$, a joint state for the combined system $AB$ can be made via a tensor product (in concrete terms a Kronecker product): $\rho_A\!\otimes\!\rho_B$.
Not all joint states of $AB$ can be written as the tensor product of a state of $A$ and a state of $B$.
This is not so different from the fact that not all joint probability distributions are the product of disjoint distributions. 
States of the form $\rho_A\!\otimes\!\rho_B$, or that are convex combinations of states of that form, are called separable\footnote{
    With the understanding of convex combinations of density matrices given above, a separable state represents the state of a system that is known to be in a product state, but with uncertainty as to which product state.
    This motivates the definition of entanglement as something more than just uncertainty over product states.
}.
Systems $A$ and $B$ are said to be \emph{entangled} if their joint state is not separable:
\begin{definition}[Entanglement]
A joint state of $d$- and $d'$-dimensional systems is \emph{entangled} if it cannot be written in the form $\sum_i\alpha_i\,\rho_i\!\otimes\!\rho'_i$ for any $\alpha_i\geq0$, $\sum_i\alpha_i=1$ and any $d$-dimensional density matrices $\rho_i$ and $d'$-dimensional density matrices $\rho'_i$.
\end{definition}

A party with access to system $A$ and a party with access to system $B$ are said to \emph{share entanglement} if $A$ and $B$ are entangled.
A concrete example of an entangled state is the so-called Bell state with density matrix
% Concrete examples of entangled states are the so-called Bell states, one of which has density matrix
\begin{equation}\label{eq:bell_state}
    \rho_\text{Bell}
    =
    \begin{bmatrix}
    \frac{1}{2}&0&0&\frac{1}{2}\\
    0&0&0&0\\
    0&0&0&0\\
    \frac{1}{2}&0&0&\frac{1}{2}\\
    \end{bmatrix},
\end{equation}
which it is possible to show cannot be written as the Kronecker product of two $2\times2$ density matrices or as the convex combination of such Kronecker products.

Similarly to states, measurements on two systems can be formally combined into a \emph{joint measurement} via a tensor product:
if $\{M_1,\cdots,M_{m_A}\}$ is a measurement on a system $A$ and $\{N_1,\cdots,N_{m_B}\}$ is a measurement on a system $B$, then $\{M_1\!\otimes\!N_1,\cdots,M_\ell\!\otimes\!N_j,\cdots,M_{m_A}\!\otimes\!N_{m_B}\}$ is a measurement on their joint system, with joint outcomes $(\ell, j)\in\left[m_A\right]\times\left[m_B\right]$.

As usual, the Born rule tells us the probabilities of obtaining these measurement outcomes:
suppose two parties, Alice and Bob, share an entangled state $\rho_{AB}$.
If Alice makes a measurement $\{M_1,\cdots,M_{m_A}\}$ on her part, and Bob makes a measurement $\{N_1,\cdots,N_{m_B}\}$ on his part, the probability that Alice obtains outcome $\ell$ and Bob obtains outcome $j$ is
\begin{equation}
    p(\ell,j)
    =
    \trace(\rho_{AB}\,M_\ell\otimes N_j).
\end{equation}
This all straightforwardly generalizes to measurements made on states shared by more than two parties:
\begin{equation}
    p(j_1,\ldots,j_n)
    =
    \trace(\rho\,M^{(1)}_{j_1}\!\otimes\ldots\otimes\! M^{(n)}_{j_n}).
\end{equation}
Here, $\rho$ is the state of a joint system of $n$ parts and each $\{M^{(k)}_j\}_j$ is a measurement made on the $k$-th part.

As we will see in Section \ref{sec:policy_spaces}, this particular form for the probabilities of measurement outcomes has implications for joint decision-making.

%%%%%%%%%%%%%%%%%%%%%%%%%%%%%%%%%%%%%%%%%%%%%%%%%%%%%%%%%%%%%
%%%%%%%%%%%%%%%%%%%%%%% Policy Spaces %%%%%%%%%%%%%%%%%%%%%%%
%%%%%%%%%%%%%%%%%%%%%%%%%%%%%%%%%%%%%%%%%%%%%%%%%%%%%%%%%%%%%
\section{Communication, Coordination, and Entangled Policies}\label{sec:policy_spaces}
Throughout the rest of the paper, we use notation common in the formalism of Dec-POMDPs~\cite{Amato2016}. For a detailed summary, please refer to Appendix~\ref{sec:decPOMDPs}. We consider $n$ agents acting in the same environment. Given a joint action space $\boldsymbol{\mathcal{A}} = \mathcal{A}_1\times\ldots\times\mathcal{A}_n$, a \emph{(joint) policy} is a function $\pi:\bigcup_{t=0}^\infty\boldsymbol{\mathcal{H}}_t \to \Delta(\boldsymbol{\mathcal{A}})$ with $\pi:\mathbf{h} \mapsto \pi(\cdot|\mathbf{h})$, from which the agents jointly sample and choose their actions $\mathbf{a}_t =    (a_{i,t})_{i\in [n]}$ based on their observation-action histories. More precisely, $\mathbf{a}_t \sim \pi(\cdot|\mathbf{h}_t)$, where $\mathbf{h}_t \in \boldsymbol{\mathcal{H}}_t$ and $\mathbf{h}_t = (h_{i,t})_{i\in[n]}$, with $h_{i,t}$ denoting the observation-action history of agent $i$ up to time step $t$.

Restrictions in communication between agents can be formalized as constraints on the class $\boldsymbol{\Pi}$ of joint policies allowed in the central task of policy evaluation. The simplest class we can consider is the class~$\boldsymbol{\Pi}_\mathsf{C}$ of \mbox{\emph{full communication}} policies, which imposes no restrictions in communication, i.e. every joint policy is allowed.

\begin{figure}[t]
    \centering
    \includegraphics[width=0.6\linewidth]{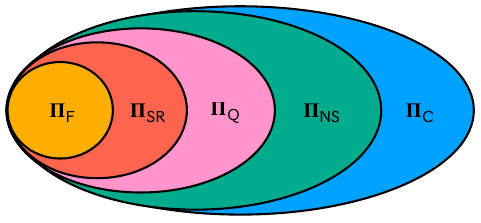}
    \caption{Hierarchy of policies. Here, $\boldsymbol{\Pi}_{\mathsf{F}}$ is the space of \emph{factorized} policies, $\boldsymbol{\Pi}_{\mathsf{SR}}$ the space of \emph{shared randomness} policies, $\boldsymbol{\Pi}_{\mathsf{Q}}$ the space of \emph{shared (quantum) entanglement} policies, $\boldsymbol{\Pi}_{\mathsf{NS}}$ the space of \emph{non-signaling} policies, and $\boldsymbol{\Pi}_{\mathsf{C}}$ the space of all joint policies.}
    \label{fig:hierarchy-of-policies}
\end{figure}

\textbf{Factorized Policies.} 
A commonly used restriction in the literature for Dec-POMDPs is given by the class~$\boldsymbol{\Pi}_{\mathsf{F}}$ of \emph{factorized} policies~\cite{Amato2016}, given by policies of the form $\pi(\mathbf{a}|\mathbf{h}) = \prod_{i=1}^n \pi_i(a_i | h_i)$
for some \emph{local policies} $\pi_1,\ldots,\pi_n$, with $\pi_i: \mathcal{H}_i \to \Delta(\mathcal{A}_i)$ for each $i\in [n]$. Policies of this form are widely adopted in cooperative MARL due to their simplicity in interpretation, ease of executing in a decentralized fashion, and due to being crucial to certain methods that exploit their functional form, such as the counterfactual multi-agent (COMA) policy gradient~\cite{foerster2024counterfactualmultiagentpolicygradients}.

\textbf{Shared Randomness.}
Next, we consider the class~$\boldsymbol{\Pi}_{\mathsf{SR}}$ of \emph{shared randomness} policies, given in its functional form by
\begin{equation}
\label{eq:shared_randomness_policies}
    \pi(\mathbf{a}|\mathbf{h}) = \sum_{x\in\mathcal{X}} q(x)\prod_{i=1}^n \pi_i(a_i|x,h_i),
\end{equation}
for some \emph{source of shared randomness} $q \in \Delta(\mathcal{X})$, with non-empty $\mathcal{X}$, and for some \emph{local policies} $\pi_1,\ldots,\pi_n$, with $\pi_i: \mathcal{X}\times\mathcal{H}_i \to \Delta(\mathcal{A}_i)$ for each $i\in [n]$. This space of policies is of relevance because it is is the broadest class of policies that can be physically implemented by decentralized non-communicating agents without quantum means. Conceptually, the shared randomness may reflect a pre-determined strategy of the agents, some shared random observation of the environment that can be used too coordinate actions, or may simply reflect a shared causal past of the agents.

For an alternative characterization of~$\boldsymbol{\Pi}_{\mathsf{F}}$ and~$\boldsymbol{\Pi}_{\mathsf{SR}}$, the following result can be established:
\begin{proposition}
$\boldsymbol{\Pi}_{\mathsf{SR}}$ is the convex hull of~$\boldsymbol{\Pi}_{\mathsf{F}}$. If~$\boldsymbol{\mathcal{H}}$ and~$\boldsymbol{\mathcal{A}}$ are finite, then $\boldsymbol{\Pi}_{\mathsf{F}}$ and $\boldsymbol{\Pi}_{\mathsf{SR}}$ can be represented as subsets of the Euclidean space $\mathbb{R}^{H\cdot A}$, with $H = |\boldsymbol{\mathcal{H}}|$ and $A = |\boldsymbol{\mathcal{A}}|$. Under such a representation, the set of deterministic policies in $\boldsymbol{\Pi}_{\mathsf{F}}$ is finite, with size $A^H$, and $\boldsymbol{\Pi}_{\mathsf{SR}}$ is a convex polytope. 
\end{proposition}

\textbf{Non-Signaling.}
Before we introduce the class of policies of primary interest in this paper, it is important to understand what ``no communication'' truly means, which we can formalize via the class~$\boldsymbol{\Pi}_\mathsf{NS}$ of \emph{non-signaling} policies. This class is defined, not in a functional form, but in terms of the following property that policies must satisfy:
\begin{equation}
    \sum_{\mathbf{a}_{\mathcal{J}}} \pi( \mathbf{a}_{\mathcal{I}}, \mathbf{a}_{\mathcal{J}} | \mathbf{h}_{\mathcal{I}}, \mathbf{h}_{\mathcal{J}} ) = \sum_{\mathbf{a}_{\mathcal{J}}} \pi( \mathbf{a}_{\mathcal{I}}, \mathbf{a}_{\mathcal{J}} | \mathbf{h}_{\mathcal{I}}, \mathbf{h}_{\mathcal{J}}' )
\end{equation}
for every partition $\{\mathcal{I},\mathcal{J}\}$ of $[n]$ and every $\mathbf{a}_{\mathcal{I}}, \mathbf{h}_{\mathcal{I}}, \mathbf{h}_{\mathcal{J}}, \mathbf{h}_{\mathcal{J}}'$. Conceptually, under non-signaling policies, the decisions taken by a group of agents can only depend on their observations and not on the observations of the remaining agents. In particular, there cannot be a mechanism for a group of agents to \emph{signal} to the remaining agents any information about what their (combined) observation-action histories are at a given point in time. 
In this sense, no communication can occur between agents as no information is allowed to be transmitted between disjoint groups of agents.

% \subsection{Shared Entanglement Policies}
\textbf{Shared Entanglement.}
\label{sec:quantum-policies}
We now introduce the primary policy class of interest: the class~$\boldsymbol{\Pi}_\mathsf{Q}$ of \emph{shared (quantum) entanglement} policies, given by policies of the form
\begin{equation}
    % \pi(\mathbf{a}|\mathbf{h}) = \mathrm{tr}\left(\rho\, \bigotimes\nolimits_{i=1}^n \!M_i(a_i | h_i) \right),
    \pi(\mathbf{a}|\mathbf{h}) = \mathrm{tr}\left(\rho\, \bigotimes_{i=1}^n \!M_i(a_i | h_i) \right),
    \label{eq:quantum-policy}
\end{equation}
for some density matrix~$\rho \in \mathbb{C}^{d\times d}$ and POVM-valued functions $M_1,\ldots,M_n$ with $M_i:\mathcal{H}_i \to \mathsf{POVM}_{m,d}$ for $i\in [n]$.

Conceptually, policies in this class can be understood as those implementable by decentralized, non-communicating agents
that each have access to different, possibly entangled, quantum systems and \emph{choose} what measurements to carry out based on locally available information only.
Of crucial importance, this space of policies is strictly larger than the space of shared randomness policies.
In operational terms, this implies that non-communicating agents with shared entanglement can implement policies with correlated behavior that is not possible by classical means such as shared randomness.
This notion, as well as other relationships between the classes of policies discussed in this section, are formalized in Proposition~\ref{prop:hierarchy}, which is illustrated in Figure~\ref{fig:hierarchy-of-policies} as a hierarchy of policy spaces. 
For additional details and discussion of so-called \emph{Bell inequality violation}, see Appendix~\ref{appendix:bell_inequality_violation}.

Additionally, a crucial aspect of quantum entanglement as a shared resource is that, in principle, it can be physically implemented in a decentralized fashion. 
More details can be found in Appendix~\ref{appendix:physical_implementation}. 
In Figure~\ref{fig:quantum-policy}, we diagrammatically illustrate the decentralized nature of a possible implementation of shared entanglement policies (leveraging Algorithm~\ref{alg:quantum-softmax}, to be introduced in Section~\ref{alg:quantum-softmax}). 

Lastly, it is worth mentioning here that popular treatments of quantum entanglement occasionally give the impression that entanglement involves a transfer of information, perhaps even ``faster than light" communication.
This is a misconception.
For us, the special property of quantum entanglement is that it allows decentralized, non-communicating actors to sample from conditional probability distributions of the form given in~\eqref{eq:quantum-policy}.
In fact, re-iterating Proposition~\ref{prop:hierarchy}, one can show that such policies satisfy the no-signaling equation, formally confirming the well-understood fact that there is no communication involved here.

\begin{proposition}
$\boldsymbol{\Pi}_\mathsf{F} \subsetneq \boldsymbol{\Pi}_\mathsf{SR} \subsetneq \boldsymbol{\Pi}_\mathsf{Q} \subsetneq \boldsymbol{\Pi}_\mathsf{NS} \subsetneq \boldsymbol{\Pi}_\mathsf{C}$.
\label{prop:hierarchy}
\end{proposition}

Finally, we highlight a useful fact that we will return to later.
Any joint policy $\pi(\mathbf{a}|\mathbf{h})$ can be written in the form
\begin{equation}\label{eq:coordinated-policy-parameterized}
    \pi(\mathbf{a}|\mathbf{h}) = \sum_\mathbf{x} q(\mathbf{x}|\mathbf{h})\prod_{i=1}^n \pi_{i}(a_i|x_i,h_i),
\end{equation}
with  $\pi_i:\mathcal{X}_i\times\mathcal{H}_i \to \Delta(\mathcal{A}_i)$ for $i\in [n]$ and $q:\boldsymbol{\mathcal{H}} \to \Delta(\boldsymbol{\mathcal{X}})$ for some non-empty sets $\mathcal{X}_1,\ldots,\mathcal{X}_n$ and $\boldsymbol{\mathcal{X}} = \mathcal{X}_1 \times \ldots \times \mathcal{X}_n$.
This form separates the joint policy into two parts: factorized individual policies $\prod_{i=1}^n \pi_{i}(a_i|x_i,h_i)$, now conditional on an additional input $x_i$, and a probability distribution $q(\mathbf{x}|\mathbf{h})$ from which the vector $\mathbf{x}=(x_1,\ldots,x_n)$ is drawn.
The additional input $x_i$ could represent the outcome of a quantum measurement, a shared random variable, or something else depending on the space of policies being represented.
Specifically, shared randomness policies are precisely those where $q(\mathbf{x}|\mathbf{h})$ does not depend on $\mathbf{h}$.
On the other hand, policies for agents sharing quantum entanglement are precisely those where $q(\mathbf{x}|\mathbf{h})$ takes the form $q(\mathbf{x}|\mathbf{h})=\trace\left(\rho\,\bigotimes_{i=1}^n M_{i}(x_i|h_i)\right)$, where $M_{i,\theta}(\,\cdot\,|h_i)$ are POVMs.
See Appendix \ref{appendix:quantum_coordinators} for a proof of this equivalence.

%%%%%%%%%%%%%%%%%%%%%%%%%%%%%%%%%%%%%%%%%%%%%%%%%%%%%%%%%%%%%
%%%%%%%%%%%%%%%%%%%%%%%% Methodology %%%%%%%%%%%%%%%%%%%%%%%%
%%%%%%%%%%%%%%%%%%%%%%%%%%%%%%%%%%%%%%%%%%%%%%%%%%%%%%%%%%%%%
\section{Learning to Coordinate via Quantum Entanglement in MARL}
\label{sec:methodology}

In this section, we derive a differentiable parameterization of policies in~$\boldsymbol{\Pi}_\mathsf{Q}$ and discuss implementation details regarding training, for both nonlocal games and Dec-POMDPs.

%%%%%%%%%%%%%%%%%%%%%%%% POVM Softmax %%%%%%%%%%%%%%%%%%%%%%%%
% \subsection{Quantum Softmax}
\subsection{Parameterizing Shared Entanglement Policies}
\label{sec:povm_softmax}

\begin{figure}[t]
    \centering
    \includegraphics[width=0.6\linewidth, trim ={1.55cm 9cm 11.5cm 4cm}, clip] {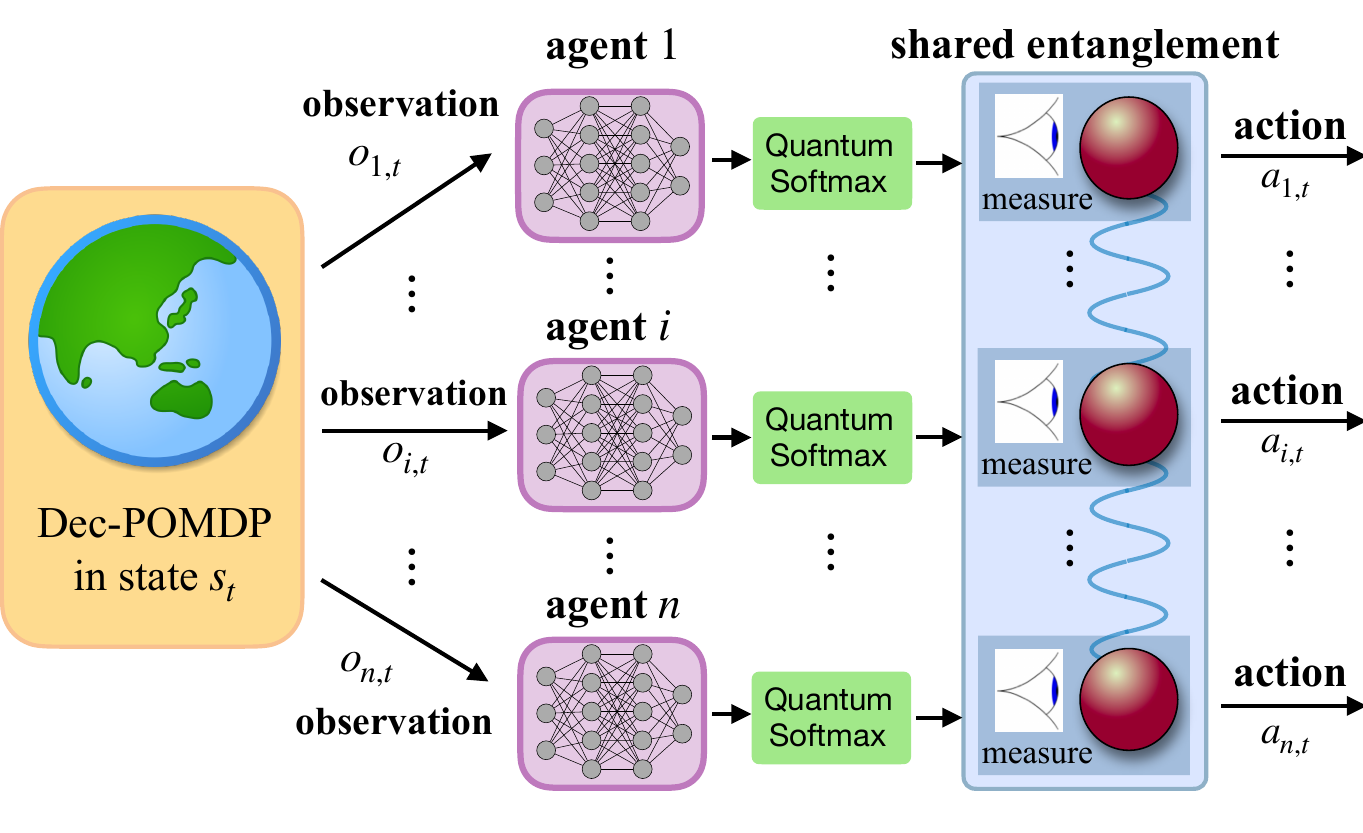}
    \caption{Decentralized and parameterized implementation of a joint policy with shared quantum entanglement.}
    %In particular, $\pi_\theta(\mathbf{a}|\mathbf{h}) = \mathrm{tr}\left[\rho\bigotimes_{i=1}^n M_{i,\theta}(a_i|h_i)\right]$.}
    \label{fig:quantum-policy}
\end{figure}

As detailed in Section \ref{sec:quantum_theory}, a measurement on a \mbox{$d$-dimensional} quantum system with $m$~possible outcomes can be formalized as a \emph{positive operator-valued measure} (POVM). We can denote the space of such POVMs as $\mathsf{POVM}_{m,d}$.
Note that $\mathsf{POVM}_{m,d}$ generalizes the probability simplex $\Delta_m = \{\mathbf{p}\in\mathbb{R}^m: \mathbf{p}\geq \mathbf{0},\, \mathbf{1}^\mathsf{T}\mathbf{p} = 1\}$, which corresponds to the case $d=1$. To parameterize probability distributions with support of finite size~$m$, it is standard practice to use $\mathsf{softmax}(\mathbf{z}) := \frac{\mathrm{exp}(\mathbf{z})}{\mathbf{1}^\mathsf{T}\mathrm{exp}(\mathbf{z})}$ as a final transformation on an $m$-dimensional vector~$\mathbf{z}$ of logits, i.e. $\mathbf{p} = \mathsf{softmax}(\mathbf{z})$. Similarly, we can transform a tuple $Z = (Z_1,\ldots,Z_m)$ of $d\times d$ complex-valued matrices into a POVM object $P = (P_1,\ldots,P_m)$. To do so, minimal adaptations need to be made on $\mathsf{softmax}(\cdot)$, which are summarized in the $\mathsf{QuantumSoftmax}(\cdot)$ function described by Algorithm~\ref{alg:quantum-softmax}.

\begin{algorithm}[t]
% \setlength{\algomargin}{-1em}
% \SetAlgoLined
\KwIn{$Z_1,\ldots,Z_m\in\mathbb{C}^{d\times d}$} % TODO: custom kw for inputs and outputs (plural)
\KwOut{$P_1,\ldots,P_m$. }
\For{$j = 1,\ldots,m$}{
    $\tilde{Z}_j = \frac{Z_j + Z_j^\mathsf{H}}{2}$ \tcp*{symmetrize}
    $R_j = \mathrm{expm}(\tilde{Z}_j)$ \tcp*{exponentiate}
}
% \For{$j = 1,\ldots,n$}{
%     {$\tilde{Z}_j = \frac{Z_j + Z_j^\mathsf{H}}{2}$;  $\ \ R_j = \mathrm{expm}(\tilde{Z}_j)$  }
% }
$S = R_1 + \ldots + R_m$\;
\For{$j = 1,\ldots,m$}{
    {$P_j = S^{-\frac{1}{2}}R_jS^{-\frac{1}{2}}$ \tcp*{normalize}}
}
\Return $P_1,\ldots,P_m$\;
\caption{$\mathsf{QuantumSoftmax}$}
\label{alg:quantum-softmax}
\end{algorithm}

The following proposition certifies that~$\mathsf{QuantumSoftmax}$ can be used as a transformation that produces POVMs:

\begin{proposition}
    For any input $Z \in (\mathbb{C}^{d\times d})^m$, Algorithm~\ref{alg:quantum-softmax} terminates and produces an output $P$ in $\mathsf{POVM}_{m,d}$. 
\label{prop:quantum-softmax-algo-valid}
\end{proposition}

Additionally, $\mathsf{QuantumSoftmax}$ is differentiable, enabling gradient-based optimization when used as a transformation. Further, any POVM can be recovered, either precisely or approximately (up to arbitrary precision), as the output of $\mathsf{QuantumSoftmax}$ for some suitable input:

\begin{proposition}
$(X,Y) \mapsto \mathsf{QuantumSoftmax}(X + \mathsf{i}Y)$ is differentiable over the real and imaginary parts, and its image is dense in~$\mathsf{POVM}_{m,d}$.
\label{prop:quantum-softmax-surjective}
\end{proposition}

Practical details about numerical stability and (automatic) differentiation (note the inverse matrix square root) can be found in Appendix~\ref{sec:differentiability-quantum-softmax}. %These claims are summarized in the following proposition:
In Figure~\ref{fig:quantum-policy}, we put everything together to obtain an end-to-end parameterization of a policy in~$\boldsymbol{\Pi}_\mathsf{Q}$.

%%%%%%%%%%%%%%%%%%%%%%%% Learning  %%%%%%%%%%%%%%%%%%%%%%%%
\subsection{Learning Quantum Entangled Strategies}
\label{sec:learning}
As communication constraints are constraints on the joint policy space, we focus on policy gradient methods, where constraints on the policy can be implemented directly.
To introduce how policy gradient methods can learn strategies that leverage quantum entanglement, we first describe a simple REINFORCE method \cite{williams1992simple} for learning strategies for nonlocal games.
We then describe a more sophisticated modified multi-agent PPO (MAPPO) \cite{yu2022surprisingeffectivenessppocooperative} method capable of finding truly sequential decision-making policies leveraging quantum entanglement.
We demonstrate these methods in the experiments of Section \ref{sec:experiments}.

\subsubsection{Policy Gradient for Nonlocal Games}
% \subsubsection{Warmup: Nonlocal Games}
\label{sec:nonlocal_warmup}
Nonlocal games are single-round, cooperative, multi-player games that allow players to determine a strategy before the start of the game but then enforce a no-communication constraint as the game is played \cite{1313847}. In each round of the game, a referee chooses questions for the players at random and sends each player their respective question. 
The players then observe their question and send an answer, after which the referee determines whether the players win or lose the game based on their collective questions and answers. As a warmup to sequential decision-making problems, we first learn strategies for nonlocal games from experience, where the players aim to find a strategy that maximizes the win probability of the nonlocal game.

In addition to maximizing the win probability of the nonlocal game, we consider adding entropy regularization, similar to \cite{williams1991function, mnih2016asynchronous, schulman2017equivalence, schulman2017proximal, ahmed2019understanding}.
For nonlocal games, there always exists an optimal classical strategy that is deterministic \cite{1313847} and therefore has zero entropy. 
Empirically, we observe that some runs converge to or linger at strategies with the classically optimal win probability. 
We hypothesize that encouraging policies with non-zero entropy helps training avoid 
% converging to these 
classically optimal deterministic strategies, where it may otherwise get stuck, and thereby discover strategies exhibiting quantum advantage. Indeed, the presence of an entropy regularization term improves the results for all nonlocal games studied in this paper, as illustrated by Table \ref{tab:nonlocal_results} in Appendix~\ref{appendix:experiment_details}. Details on nonlocal games, our model architecture and the training procedure are provided in Appendix~\ref{appendix:nonlocal_games}.

\subsubsection{Multi-Agent PPO for Sequential Decision-Making}\label{sec:modified_ppo}
Nonlocal games as usually understood are single-round games where actions do not affect future rounds. 
Optimal strategies are more difficult to find for Dec-POMDPs where sequential decision-making between non-communicating agents is required. 
This is the domain of reinforcement learning algorithms.

Designing reinforcement learning algorithms that find policies satisfying the various no-communication constraints outlined in Section \ref{sec:policy_spaces} requires some thought.
As these constraints are restrictions on the form the joint policy can take it is convenient to focus on policy gradient methods, where the policy is directly parametrized and we can choose a parametrization that itself enforces the constraint.

Many multi-agent RL algorithms (MAPPO, for example) train the individual policies rather than the joint policy.
For us the no-communication constraints at the level of the joint policy couple the individual policies of the agents, such that training individual policies may seem unsuitable at first.
We make use of the fact that any joint policy can be written as in~\eqref{eq:coordinated-policy-parameterized}, $\pi_\theta(\mathbf{a}|\mathbf{h})=\sum_\mathbf{x} q_\theta(\mathbf{x}|\mathbf{h})\prod_{i=1}^n \pi_{i,\theta}(a_i|x_i,h_i)$, with different communication constraints now encoded in the form that the distribution $q_\theta(\mathbf{x}|\mathbf{h})$ takes.
This ansatz separates the joint policy into actor models for individual policies $\pi_i(a_i|x_i,h_i)$ and a ``coordinator" model for the distribution $q(\mathbf{x}|\mathbf{h})$.
This allows us to use some of the machinery of algorithms that train individual policies $\pi_i$ as long we are able to make two modifications: we include an additional model for $q_\theta(\mathbf{x}|\mathbf{h})$, and we allow the individual policies $\pi_i$ to depend on an additional input $x_i$ sampled from $q_\theta(\mathbf{x}|\mathbf{h})$.
This additional input $x_i$, which we call the advice, can be treated similarly to an observation of agent $i$ and considered part of the trajectory together with states, observations, and actions.

As a concrete example we implement a MAPPO \cite{yu2022surprisingeffectivenessppocooperative} algorithm modified to include a coordinator model in addition to the usual actor and critic models. 
The surrogate objective function of PPO \cite{schulman2017proximal} is chosen so that its gradient, before clipping, is the policy gradient estimate
$ % \begin{equation}
\nabla_\theta J(\theta)
=
\mathbb{E}_{\tau}\left[
    \sum_t
    \nabla_{\theta}\log\pi_{\theta}(a_t|h_t)
    A_t
\right],
$ % \end{equation}
where $A_t$ is (an estimate of) the advantage.
Replacing the policy in this policy gradient estimate with the form $\pi(\mathbf{a}_t|\mathbf{h}_t)=\sum_{\mathbf{x}_t}q(\mathbf{x}_t|\mathbf{h_t})\prod_i\pi_i(a_{it}|x_{it},h_{it})$ suggests a surrogate objective with two terms: a term for the actor models for the policies $\pi_{i}(a_{i}|x_i,h_{i})$,
\begin{equation}
    \label{eq:actor_surrogate}
    L_\text{act}(\theta)
    =
    \mathbb{E}_\tau\bigg[\sum_t\sum_i
    \frac{
        \pi_{i\theta}(a_{it}|x_{it},h_{it})
    }{
        \pi_{i\theta_0}(a_{it}|x_{it},h_{it})
    }
    A_t
    \bigg],
\end{equation}
and a term for the coordinator model
\begin{equation}
    \label{eq:coord_surrogate}
    L_\text{coord}(\theta)
    =
    \mathbb{E}_\tau\bigg[\sum_t
    \frac{
        q_\theta(\mathbf{x}_t|\mathbf{h}_t)
    }{
        q_{\theta_0}(\mathbf{x}_t|\mathbf{h}_t)
    }
    A_t
    \bigg].
\end{equation}
The expectations here are over trajectories that are understood to now include sampled values of $\mathbf{x}_t$ in addition to sampled states, observations, and actions.
See Appendix \ref{appendix:surrogate_objective_function}
These terms can each then be clipped as in PPO.

This together with an unmodified centralized critic constitute our modified MAPPO algorithm.
With the addition of the coordinator model and corresponding coordinator objective term, the algorithm will find solutions that respect whatever constraints are encoded in the form that $q(\mathbf{x}|\mathbf{h})$ takes.
For example, to learn policies that utilize shared quantum entanglement for each joint action, we can model $q(\mathbf{x}|\mathbf{h})$ to have the form $q_\theta(\mathbf{x}|\mathbf{h})=\trace\left(\rho_\theta\bigotimes_{i=1}^n M_{i,\theta}(x_i|h_i)\right)$, with the POVMs being parametrized by the $\mathsf{QuantumSoftmax}$ layer described in Section \ref{sec:povm_softmax}, for example.
To learn joint policies where only shared randomness is shared between the agents, we could instead model $q(\mathbf{x}|\mathbf{h})$ to have the simpler form $q(\mathbf{x}|\mathbf{h})=q(\t{\mathbf{x}})\delta_{\mathbf{x},\t{\mathbf{x}}}$ where the random variables~$\{x_i\}_{i\in [n]}$ are equal and do not depend on the history~$\mathbf{h}$.

When deploying a policy learned by this method, the advice $\mathbf{x}$ is first sampled from the coordinator distribution $q(\mathbf{x}|\mathbf{h})$ and then each agent samples its action $a_i$ from its individual advised policy $\pi_i(a_i|x_i,h_i)$
As explained in Section \ref{sec:policy_spaces} when the form of $q(\mathbf{x}|\mathbf{h})$ is chosen to reflect a no-communication constraint, sampling from it is possible, with the appropriate resources, in a decentralized fashion.

%%%%%%%%%%%%%%%%%%%%%%%%%%%%%%%%%%%%%%%%%%%%%%%%%%%%%%%%%%%%%
%%%%%%%%%%%%%%%%%%%%%%%% Experiments %%%%%%%%%%%%%%%%%%%%%%%%
%%%%%%%%%%%%%%%%%%%%%%%%%%%%%%%%%%%%%%%%%%%%%%%%%%%%%%%%%%%%%
\section{Experiments}
\label{sec:experiments}

\subsection{Nonlocal Games}
We first evaluate the performance of our framework on four nonlocal games where strategies with quantum advantage have been discovered previously in the literature. Specifically, we learn strategies for the CHSH game \cite{original_chsh, 1313847}, the GHZ game \cite{vaidman1999variations} and two forms of the rendezvous game \cite{PhysRevA.109.042201}. The details of these games can be found in Section \ref{appendix:nonlocal_games} of the Appendix. We visualize the results for the CHSH game in Figure \ref{fig:chsh_performance}, which plots the win probabilities of our learned strategies with and without entropy regularization during training. As hypothesized, we see that entropy regularization does indeed help find quantum advantage, as all of the runs with entropy regularization find quantum advantage, whereas 3 of the runs without entropy regularization get stuck at a strategy with a win probability that is classically optimal. The benefits of entropy regularization are observed for all of the nonlocal games, as every run with entropy regularization finds a strategy with quantum advantage, whereas there is always at least one run without entropy regularization that cannot surpass the classically optimal win probability. Results for the other nonlocal games are detailed in Appendix~\ref{appendix:experiment_details}.

\begin{figure}[t]
\centering
\includegraphics[width=0.5\linewidth, trim ={0.5cm 0cm 0cm 0cm}]{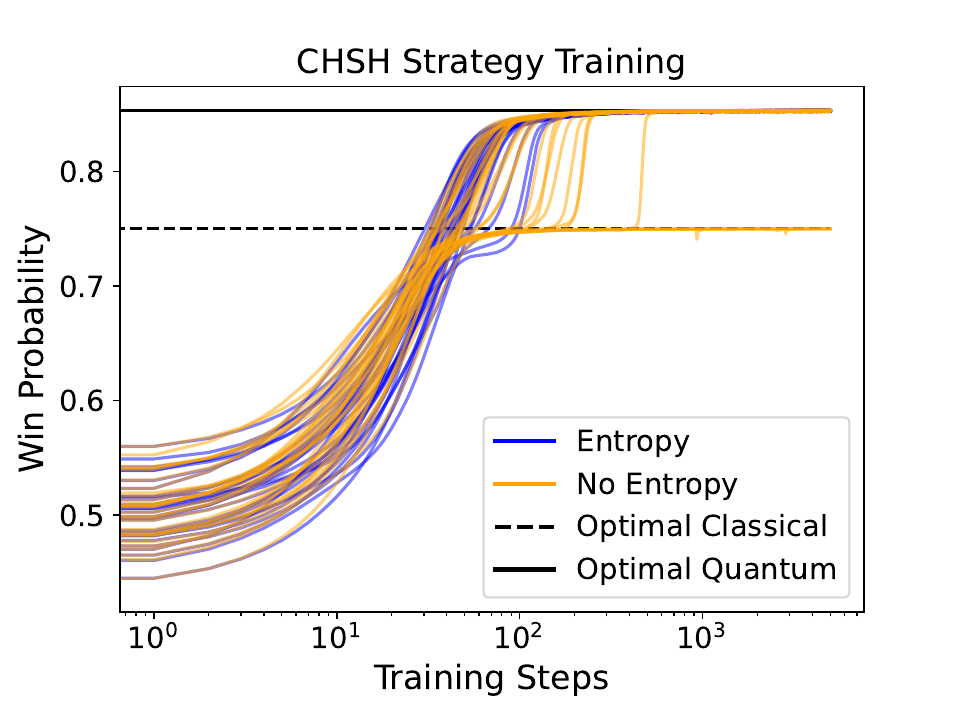}
    \caption{Win probabilities of learned strategies for the CHSH game during training with and without entropy regularization.}
    \label{fig:chsh_performance}
\end{figure}

\subsection{Multi-Agent Sequential Decision-Making}
In this section, we illustrate our experiments and results when using our proposed framework to learn strategies from experience in a sequential decision making problem analyzed in previous works which provably admits quantum advantage. 

\subsubsection{Multi-Router Multi-Server Queueing}
We consider a decentralized routing problem introduced by~\citet{da2025entanglement, dasilva2026entanglement}. 
In this problem, two decentralized decision-makers (routers) route customers servers based only on the service time of a customer's request. 
When routed to a server, a customer joins the server's queue, and the server provides service to customers for as long as the queue is nonempty.
An idle server (one with empty queue) works on a ``baseline task," with efficiency that increases with the uninterrupted idle time. 
The common goal of the routers is to maximize the throughput on this baseline task while guaranteeing a bounded average waiting time for the customers. 
The key decision-making challenge for a router lies in not knowing the service times for customers arriving at the other router. 
The setup is illustrated in Fig.~\ref{fig:server_router_fig}. 
The authors~\citet{dasilva2026entanglement} analyzed this problem by modeling inter-arrival and service times as exponential random variables and showed, using results from queueing theory, that entanglement shared between the routers can provide quantum advantage in the decentralized decision-making.
\begin{figure}[t]
    \centering
    \includegraphics[width=0.6\linewidth]{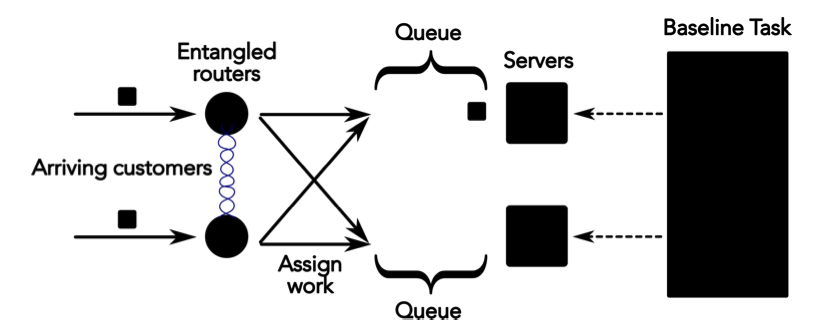}
    \caption{The multi-router queueing problem}
    \label{fig:server_router_fig}
\end{figure}

To demonstrate our MARL framework, we model this task as a Dec-POMDP and learn suitable policies from experience, eschewing the analytical queueing theory results that underlie the discovery of quantum advantage by~\citet{dasilva2026entanglement}.
See Appendix~\ref{appendix:server-router} for a complete description of the problem as a MARL environment.

The agents have the dual objective of increasing throughput on the baseline task and reducing customer wait time.
We treat the dual objective by using the throughput as the RL reward and use the PPO-Lagrangian \cite{Achiam2019BenchmarkingSE} framework with \mbox{PID-controlled} Lagrange multiplier \cite{stooke2020responsivesafetyreinforcementlearning} to enforce an upper bound on the wait time.

\begin{figure}[!ht]
\centering
\includegraphics[width=0.5\linewidth, trim ={0.5cm 0cm 0cm 0cm}]{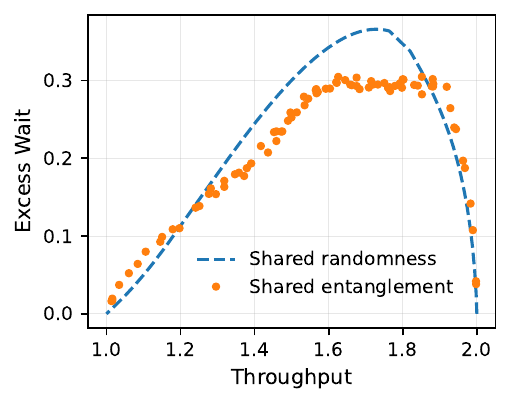}
    \caption{
        Wait time in excess of the optimal wait time obtainable when communication is allowed. 
        Lower is better.
        Learned strategies with shared quantum entanglement (orange dots) outperform the theoretical best for non-entangled strategies (blue line) for most values of the throughput.
    }
    \label{fig:server_router_perf}
\end{figure}

Figure~\ref{fig:server_router_perf} shows the performance of our algorithm.
The results are in terms of the difference between the obtained customer wait time and the theoretically optimal wait time with communication allowed, i.e.\ we plot the excess wait time directly attributable to the no-communication constraint.
For most values of the throughput, the models learn shared entanglement strategies (orange dots) that provide lower wait time than the known optimal wait times for shared randomness strategies (blue line).
The presence of shared quantum entanglement ameliorates the effect of the communication constraint.
See Appendix \ref{appendix:server-router} for details on the experiment.

By these results, we confirm the finding of \citet{dasilva2026entanglement} by alternate means.
Since our approach does not rely on queueing theory results particular to the problem, our findings can in principle be extended to other, more general, examples.

%%%%%%%%%%%%%%%%%%%%%%%%%%%%%%%%%%%%%%%%%%%%%%%%%%%%%%%%%%%%
%%%%%%%%%%%%%%%%%%%%%%%% Conclusion %%%%%%%%%%%%%%%%%%%%%%%%
%%%%%%%%%%%%%%%%%%%%%%%%%%%%%%%%%%%%%%%%%%%%%%%%%%%%%%%%%%%%
% \section{Discussion and Conclusion}
\section{Conclusions and Future Work}
\label{sec:conclusion}

In this work, we introduced a MARL framework for learning joint policies that leverage quantum entanglement as a coordination resource.
We describe a parametrization of entangled policies via learnable POVMs ($\mathsf{QuantumSoftmax}$) and demonstrate how to incorporate this parametrization into policy gradient methods by training both local actor models and a coordinator model with an appropriately modified MAPPO objective.
Empirically, this machinery discovers strategies exhibiting quantum advantage in communication-free settings, including in an illustrative Dec-POMDP.

Several directions for future work present themselves:
% This suggests several directions for further research:
\begin{itemize}
    \item  First and most pressing is the need to characterize, theoretically or empirically, settings where entanglement can improve decentralized sequential decision-making.
    \item Second, our framework limits shared entanglement to use by agents at the same time step, 
    only allowing entanglement to directly coordinate actions within a time step.
    This is not the most general way agents could consume entanglement. 
    To our knowledge, there is no work on e.g.\ multi-round games where players share entanglement \emph{between} rounds, but believe that generalization along this line is crucial for realistic applications, as decentralized agents are likely to encounter the need for \emph{asynchronous coordination}. 
    %How generic is quantum advantage in complex MARL environments? % When can quantum advantage be ruled out?
    \item Third, there is a need to incorporate realistic constraints, such as imperfect state preparation, restricted measurement sets, or small dimensionality, directly into the learning problem.
    See Appendix \ref{appendix:physical_implementation}.
    \item Finally, it is known in classical multi-agent systems that shared randomness can yield more efficient finite-state controller solutions \cite{Amato2016}.
    It is natural to ask whether shared entanglement can reduce the size of finite-state controllers further.
\end{itemize}

This work was inspired by the recent works of \citet{ding2024coordinating} and \citet{ding2025quantumnonlocalitylatencyconstraints}, which call attention to the communication barrier presented by latency in high-frequency trading and suggest the in-principle usefulness of shared quantum entanglement in coordinating trades.
This potential for commercial usefulness is in stark contrast with the erstwhile reputation of Bell inequality violation as a phenomenon of conceptual or perhaps philosophical importance.
This shift in thinking tracks progress in the quantum technology that makes distributing entanglement possible.
As the technology matures, the questions become more pressing: how generic is quantum advantage, and how discoverable is that quantum advantage? 
Multi-agent reinforcement learning presents itself as a useful tool for addressing these questions, with this work making a crucial step in this direction.

\subsubsection*{Acknowledgments}
We thank Douglas Hamilton for his support throughout the project and Sima Noorani for contributing to discussions and some early code.
We also thank Francisco Ferreira da Silva for crucial discussions related to the multi-router queueing problem and feedback on the draft.

\bibliography{bibliography_ICML}

\appendix
\onecolumn
\subsubsection*{Appendix}

\section{Bell Inequality Violation}
\label{appendix:bell_inequality_violation}
Consider two random variables $A$ and $B$ with no causation between them. 
All causes of $A$ and $B$ can be divided into shared causes and unshared causes.
Denote the shared causes of $A$ and $B$ by $\Lambda$, denote the causes of $A$ unshared with $B$ by $X$, and denote the causes of $B$ unshared with $A$ by $Y$.
The situation is described by the causal diagram in figure \ref{fig:causal_diagram}.
\begin{figure}[!ht]
\begin{center}
\includegraphics[width=0.35\linewidth]{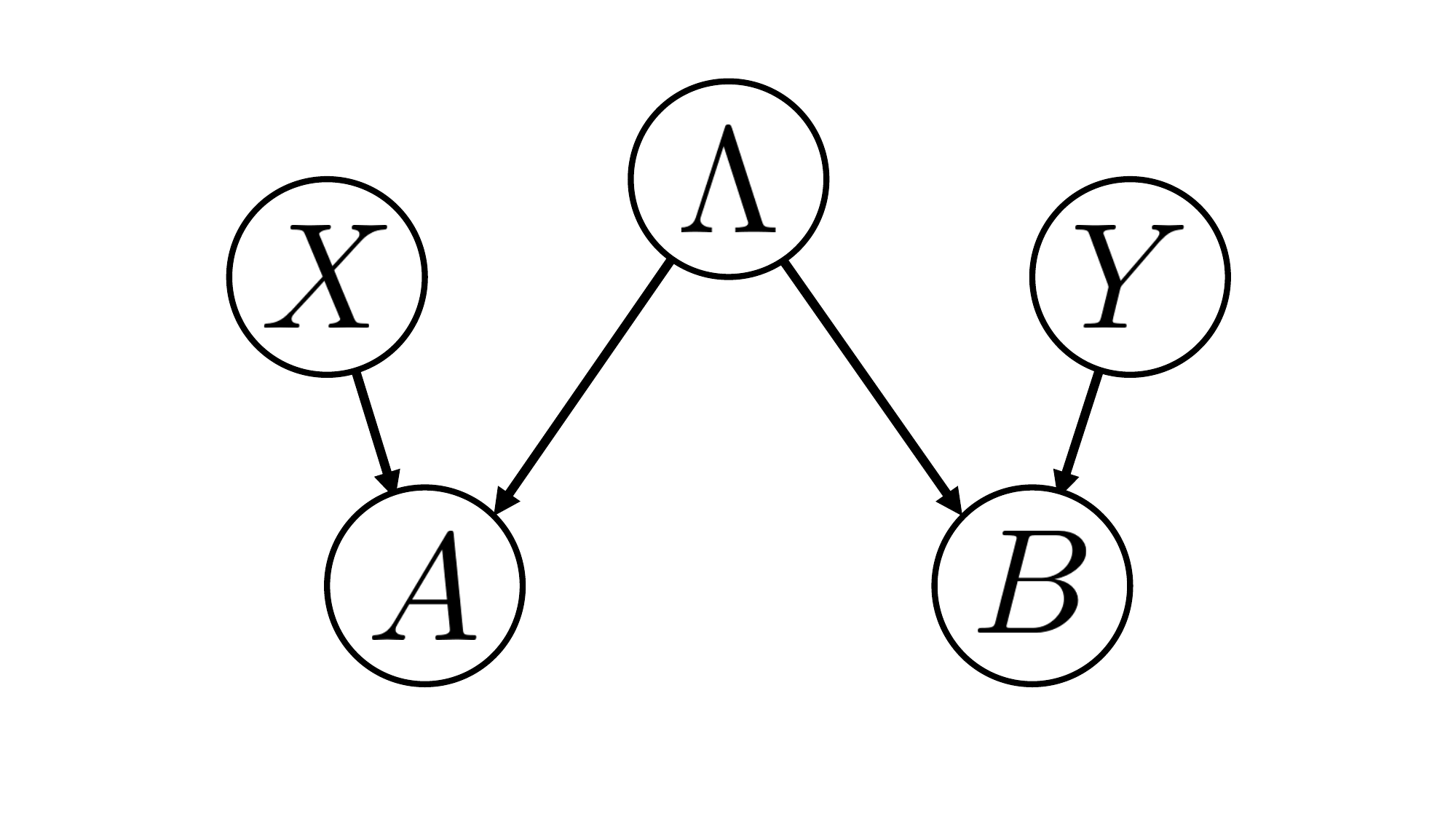}
\end{center}
\caption{Causal relationships between $A$, $B$, $X$, $Y$, and $\Lambda$}
\label{fig:causal_diagram}
\end{figure}
In general, it is tempting to understand $A$, $B$, $X$, $Y$, and $\Lambda$ as random variables and interpret the causal diagram as a Bayesian network describing the form their joint probability distribution must take, in this case $p(a,b,x,y)=p(x)p(y)p(\lambda)p(a|x,\lambda)p(b|y,\lambda)$.
In particular, the joint distribution over $A$ and $B$ conditional on $X$ and $Y$ would be
\begin{equation}\label{eq:hidden_variable}
    p(a, b|x, y)
    =
    \sum_\lambda p(\lambda)p(a|x,\lambda)(b|y,\lambda).
\end{equation}
Again, these forms would follow from the conditional independence implied by the causal relations between $A$, $B$, $X$, $Y$, and $\Lambda$, if we were to consider them as random variables.
This  would be a natural assumption to make in nearly all real-world contexts typically encountered. 
The assumption can fail, fascinatingly, when A and B have entangled quantum states in their shared past.
This situation has no description in terms of some shared random variable $\Lambda$.
Suppose Alice and Bob have respective quantum systems that are in a joint state $\rho$.
Further suppose Alice and Bob both make measurements on their respective quantum systems, and that Alice's choice of measurement is $x$ sampled from a random variable $X$, and Bob's choice of measurement is $y$ sampled from a random variable $Y$.
Let $A$ be the random variable describing the outcome of Alice's measurement and $B$ be the random variable describing the outcome of Bob's measurement.
Then the probability of joint outcome $(a, b)$, conditional on $X$ and $Y$, is given by the Born rule as
\begin{equation}\label{eq:born_rule_bell_inequality}
    p(a, b|x, y)
    =
    \trace\big(\rho\,M(a|x)\!\otimes\! N(b|y)\big),
\end{equation}
where $M(\,\cdot\,|x)$ is the measurement corresponding to Alice's choice $x$, and $N(\,\cdot\,|y)$ is the measurement corresponding to Bob's choice $y$.
There are states $\rho$ and appropriate families of measurements $\left\{M(\,\cdot\,|x)\right\}_x$ and $\left\{N(\,\cdot\,|y)\right\}_y$ such that the RHS of \eqref{eq:born_rule_bell_inequality} cannot be written in the form $\sum_\lambda p(\lambda)p(a|x,\lambda)(b|y,\lambda)$ for \emph{any} choice of $p(\lambda)$, $p(a|x,\lambda)$, and $p(b|y,\lambda)$.
In fact the RHS of \eqref{eq:born_rule_bell_inequality} is a strictly more general form.
The failure of the RHS of \eqref{eq:born_rule_bell_inequality} to take the form \eqref{eq:hidden_variable} is known to physicists as Bell inequality violation \cite{RevModPhys.86.419}.
The simplest explicit example of a state and measurements where this happens is the CHSH game, which we discuss in Appendix \ref{appendix:chsh}.

Alice and Bob can choose their measurements randomly and independently, without any communication between them, and the state $\rho$ can be prepared and shared beforehand.
Then the shared causes of random variables $A$ and $B$ are entirely captured in the state $\rho$.
As we've seen it is not always possible to represent these shared causes as a random variable and make the usual assumptions one would expect from the causal structure.
This is an example where \emph{physical} causality does not have the same implications that causality between random variables is usually understood to have.
The usual, intuitive ``rules" connecting causation and the resulting form of joint distributions are only approximations contingent on the macroscopic nature of the quantities represented by the variables. 
Though unintuitive, this fact has been experimentally verified repeatedly since it was first described theoretically by John Stewart Bell in 1964 \cite{bell1964einstein}.
Some of the work done to experimentally verify this fact was the subject of the 2022 Nobel Prize in Physics awarded to physicists Alain Aspect, John Clauser, and Anton Zeilinger \cite{NobelPrizePhysics2022}.

This feature of quantum entanglement has practical significance.
Reinterpret $X$ and $Y$ as the observations of two agents, $A$ and $B$ as their actions, and $\Lambda$ as their shared history.
Absent shared entanglement, the only possible policies the agents can execute, given no communication between them, will be of the form
\begin{equation}\label{eq:hidden_variable_policy}
    \pi(a,b|x,y)
    =
    \sum_\lambda p(\lambda)p(a|x,\lambda)(b|y,\lambda).
\end{equation}
With a shared entangled state, however the policy can take the form
\begin{equation}\label{eq:born_rule_policy}
    \pi(a, b|x, y)
    =
    \trace\left(\rho\,M(a|x)\!\otimes\! N(b|y)\right),
\end{equation}
which is strictly more general.
More generally, actions $a_i$ of non-communicating agents based on unshared observations $o_i$ need not be restricted to being drawn from policies of the form $\pi(a_1\cdots a_n|o_1\cdots o_n)=\sum_\lambda p(\lambda)\pi_1(a_1|o_1,\lambda)\cdots\pi_n(a_n|o_n,\lambda_n)$.
A larger space of policies is possible when quantum entanglement is shared between the agents.

We present here a simple geometric description of the spaces of shared randomness and shared entanglement policies defined by the forms in \eqref{eq:quantum-policy} and \eqref{eq:shared_randomness_policies} respectively.
For the sake of the following discussion, assume that the joint agent history $\mathbf{h}$ takes value in a finite set $\mathcal{H}$, and the joint action $\mathbf{a}$ takes value in a finite set $\mathcal{A}$.
Joint policies $\pi(\mathbf{a}|\mathbf{h})$ are then points in $\mathbb{R}^{|\mathcal{A}||\mathcal{H}|}$.
Consider the factorizable, deterministic joint policies, i.e.\ those of the form $\pi(\mathbf{a}|\mathbf{h})=\prod_i\delta_{a_i,f_i(h_i)}$.
It is not hard to see that the space of shared randomness policies is the convex hull of these factorizable, deterministic policies, and is thus a polytope in $\mathbb{R}^{|\mathcal{A}||\mathcal{H}|}$.
The space of policies allowed by quantum entanglement contains the space of shared randomness policies.
It is also convex, though not a polytope.
Facets of the polytope of shared randomness policies are described by linear inequalities, called by physicists Bell inequalities.
A quantum entangled policy that cannot be written as the RHS of \eqref{eq:shared_randomness_policies} is said to exhibit ``Bell inequality violation."
Again, see Appendix \ref{appendix:chsh} for an explicit example of a policy using quantum entanglement that cannot be implemented with merely shared randomness.

\section{Dec-POMPDs in more detail}\label{sec:decPOMDPs}
In \emph{fully cooperative} MARL, a power formalism is the \emph{decentralized partially observable Markov decision process} (Dec-POMDP)~\cite{Amato2016}, given by a tuple $\mathcal{M} = \left(\mathcal{S}, \{\mathcal{A}_i\}_{i\in [n]}, \{\mathcal{O}_i\}_{i\in [n]}, R, P, O, \gamma\right)$. At each time step~$t$, the environment is at a state $s_t\in\mathcal{S}$ and, simultaneously, each of agent $i\in [n]$ observes $o_{i,t} \in \mathcal{O}_i$ and immediately takes an action $a_{i,t} \in \mathcal{A}_i$. Subsequently, the environment transitions into a new state $s_{t+1}$ at the next time step, according to the transition probability function $P:\mathcal{S}\times\boldsymbol{\mathcal{A}}\to\Delta(\mathcal{S})$, with $P:(s,\mathbf{a})\mapsto P(\cdot|s,\mathbf{a})$. More precisely, $s_{t+1} \sim P(\cdot|s_t,\mathbf{a}_t)$, where $\mathbf{a}_t \in \boldsymbol{\mathcal{A}}$, with $\mathbf{a}_t = (a_{i,t})_{i\in [n]}$ and $\boldsymbol{\mathcal{A}} = \mathcal{A}_1\times\ldots\times\mathcal{A}_n$. Upon a state transition, the environment will produce a reward $r_t$ that is shared across all agents, according to the reward function $R:\mathcal{S}\times\boldsymbol{\mathcal{A}}\times\mathcal{S} \to \mathbb{R}$. More precisely, $r_t = R(s_t,\mathbf{a}_t,s_{t+1})$. The joint observation $\mathbf{o}_t \in \boldsymbol{\mathcal{O}}$, with $\mathbf{o}_t = (o_{i,t})_{i\in [n]}$ and $\boldsymbol{\mathcal{O}} = \mathcal{O}_1 \times \ldots \times \mathcal{O}_n$, depends on the current state $s_t$ as well as the joint action $\mathbf{a}_{t-1}$ that led to the current state, according to the observation function $O:\mathcal{S}\times\boldsymbol{\mathcal{A}} \to \Delta(\boldsymbol{\mathcal{O}})$. More precisely, $\mathbf{o}_t \sim O(\cdot|s_t,\mathbf{a}_{t-1})$.

A \emph{(joint) policy} is a function $\pi:\bigcup_{t=0}^\infty\boldsymbol{\mathcal{H}}_t \to \Delta(\boldsymbol{\mathcal{A}})$ with $\pi:\mathbf{h} \mapsto \pi(\cdot|\mathbf{h})$, from which the agents jointly sample and choose their actions based on their observation-action histories. More precisely, $\mathbf{a}_t \sim \pi(\cdot|\mathbf{h}_t)$, where $\mathbf{h}_t \in \boldsymbol{\mathcal{H}}_t$ and $\mathbf{h}_t = (h_{i,t})_{i\in[n]}$, with $h_{i,t}$ denoting the observation-action history of agent $i$ up to time step $t$. For a given joint policy~$\boldsymbol{\pi}$ and a given initial state $s\in\mathcal{S}$, the value function
\begin{equation}
    V^{\boldsymbol{\pi}}(s) := \mathbb{E}_{\boldsymbol{\pi}}\left[\sum_{t=0}^\infty \gamma^t\,R(s_t,\mathbf{a}_t,s_{t+1}) \,|s_0 = s\, \right]
    \label{eq:value-function}
\end{equation}
quantifies the expected total discounted reward, where \mbox{$0\leq \gamma < 1$} denotes the \emph{discount factor}. Given a class $\boldsymbol{\Pi}$ of joint policies, \emph{policy optimization} is the central task we will consider, which consists of approximately finding
\begin{equation}
    \boldsymbol{\pi}^\star \in \argmax_{\boldsymbol{\pi}\in\boldsymbol{\Pi}} V^{\boldsymbol{\pi}}(\mu),
    \label{eq:policy-optimization}
\end{equation}
where $V^{\boldsymbol{\pi}}(\mu) := \mathbb{E}_{s\sim\mu}[V^{\boldsymbol{\pi}}(s)]$, with $\mu \in \Delta(\mathcal{S})$ denoting some initial state distribution.

\section{Proofs}
\label{appendix:proofs}
\subsection{Proof of Proposition \ref{prop:quantum-softmax-algo-valid}}
\begin{proof}
Let $Z \in(\mathbb{C}^{d\times d})^n$ be an arbitrary input to Algorithm~\ref{alg:quantum-softmax} and let $P$ denote the tentative output. By definition, we have $P_j = S^{-\frac{1}{2}}R_j S^{-\frac{1}{2}}$ with $S = \sum_{j=1}^n R_j$ and $R_j = \mathrm{exp}\left(\frac{Z_j + Z_j^\mathsf{H}}{2}\right)$. Now, note that $\tilde{Z}_j := \frac{Z_j + Z_j^\mathsf{H}}{2}$ is Hermitian, and thus $R_j = \mathrm{exp}(\tilde{Z}_j)$ is also Hermitian. Further, since $\tilde{Z}_j$ is Hermitian, it must have real eigenvalues, say $\lambda_1(\tilde{Z}_j),\ldots,\lambda_d(\tilde{Z}_j) \in \mathbb{R}$. Subsequently, $R_j = \mathrm{exp}(\tilde{Z}_j)$ must have eigenvalues $\mathrm{exp}(\lambda_1(\tilde{Z}_j)),\ldots,\mathrm{exp}(\lambda_d(\tilde{Z}_j))$, which are clearly real and strictly positive. Therefore, $R_1,\ldots,R_n \succ 0$. Clearly then, $S = R_1+\ldots+R_n\succ 0$, which makes $S^{-\frac{1}{2}}$ well defined.

All that remains is to show now is that $P_1,\ldots,P_n \succeq 0$, with $P_1 + \ldots + P_n = I_d$. In fact, we will show a stronger claim, namely that $P_1,\ldots,P_n \succ 0$, in addition to $P_1 + \ldots + P_n = I_d$. First, note that, given any $z\in\mathbb{C}^d\setminus\{0\}$ and $j\in\{1,\ldots,n\}$, we have
\begin{align*}
    z^\mathsf{H} P_j z &= z^\mathsf{H}\big(S^{-\frac{1}{2}} R_j S^{-\frac{1}{2}}\big) z \\
    &= \big(S^{-\frac{1}{2}}z\big)^\mathsf{H} R_j \big(S^{-\frac{1}{2}}z\big) \tag{$S = S^\mathsf{H}$} \\
    &= w^\mathsf{H} R_j w \tag{$w = S^{-\frac{1}{2}}z$} \\
    &>0,
\end{align*}
where the inequality can be justified as follows: \emph{(i)} $R_j \succ 0$ as previously established, and \emph{(ii)} $S \succ 0$ as previously established, which ensures that $S^{-\frac{1}{2}}$ is positive definite and thus invertible, and therefore $w \neq 0$. Since $z$ and $j$ were arbitrary, we conclude that $P_1,\ldots,P_n \succ 0$.

Lastly,
\begin{align*}
    \sum_{j=1}^n P_j &= \sum_{j=1}^nS^{-\frac{1}{2}} R_j S^{-\frac{1}{2}} \\
    &= S^{-\frac{1}{2}} \left(\sum_{j=1}^nR_j\right) S^{-\frac{1}{2}} \\
    &= S^{-\frac{1}{2}} S S^{-\frac{1}{2}} \\
    &= I_d,
\end{align*}
which completes the proof.
\end{proof}

\section{Nonlocal Games}
\label{appendix:nonlocal_games}
\subsection{Formal Definition}
For a nonlocal game of $n$ players, let $\boldsymbol{\mathcal{O}} = \mathcal{O}_1\times\ldots\times\mathcal{O}_n$ for finite, non-empty $\mathcal{O}_i$ define the joint observation space corresponding to the possible combinations of questions for all of the players. 
In addition, let $\boldsymbol{\mathcal{A}} = \mathcal{A}_1\times\ldots\times\mathcal{A}_n$ for finite, nonempty $\mathcal{A}_i$ define the joint action space corresponding to the possible combinations of answers for all of the players. 
The nonlocal game is then defined as $G = (V, \mu)$ where $V$ denotes a predicate function where $V(\mathbf{a} | \mathbf{o})$ is 1 if the response $\mathbf{a} \in \boldsymbol{\mathcal{A}}$ is correct given the question $\mathbf{o} \in \boldsymbol{\mathcal{O}}$ and 0 otherwise, and $\mu$ denotes the distribution over the questions $\mathbf{o}$ that the referee asks to the players, i.e., $\mathbf{o} \sim \mu$. 
The players then seek to find a strategy $\pi_\theta(\cdot|\mathbf{o})$, parametrized by $\theta$, that maximizes the probability of winning the game. 

Given a strategy $\pi_\theta$, the probability that the players win the game $G$ is
\begin{equation}
    \label{eq:non_local_game_obj}
    \begin{aligned}
    J(\theta) &= \mathbb{E}_{\mathbf{o} \sim \mu} \mathbb{E}_{\mathbf{a} \sim \pi_{\theta}(\cdot | \mathbf{o})}[V(\mathbf{a}|\mathbf{o})] \\
    % &= \sum_o \mu(o) \sum_a \pi_{\theta}(a|o) V(a|o)
    \end{aligned}
\end{equation}

Traditionally, the literature on nonlocal games assumes that the players know the distribution over the observations $\mu$ as well as the predicate function $V$ \cite{1313847}. 
However, we will show that this knowledge is not required and that we can learn strategies for nonlocal games from samples (i.e.\ with the referee acting as a black box oracle returning the value of the predicate function $V(\mathbf{a}|\mathbf{o})$ after each round of the game).
In particular, we use the score function estimator from the REINFORCE algorithm \cite{williams1992simple} for the gradient of the strategy parameters and then perform stochastic gradient ascent. 
Specifically, we estimate the gradient of the winning probability in \eqref{eq:non_local_game_obj} with respect to the strategy parameters as
\begin{equation}
    \label{eq:non_local_gradient_estimate}
    \widehat{\nabla_\theta J(\theta)} = \frac{1}{N}\sum_{j=1}^{N}V(\mathbf{a}^{(j)}|\mathbf{o}^{(j)})\nabla_\theta \log\pi_\theta(\mathbf{a}^{(j)}|\mathbf{o}^{(j)})
\end{equation}
given $N$ i.i.d. samples $\{(\mathbf{o}^{(j)},\mathbf{a}^{(j)})\}_{j\in[N]}$ from the joint distribution $p_{\theta}(\mathbf{o}, \mathbf{a}) = \mu(\mathbf{o})\pi_{\theta}(\mathbf{a}|\mathbf{o})$, i.e., playing the game $N$~times. 
We provide a simple derivation of this estimator in \ref{appendix:nonlocal_reinforce_proof} for completeness.

\subsection{Entropy Regularization}
\label{appendix:entropy_regularization}
We consider adding a conditional entropy term $\mathcal{H}(\theta)$ to the winning probability of the game, resulting in a new objective function $J_\text{ent}(\theta) = J(\theta) + \alpha\mathcal{H}(\theta)$ with a regularization coefficient $\alpha \geq 0$. The conditional entropy of a strategy $\pi_\theta$ is defined as 
\begin{equation}
    \mathcal{H}(\theta) = -\mathbb{E}_{\mathbf{o} \sim \mu} \mathbb{E}_{\mathbf{a} \sim \pi_\theta(\cdot|\mathbf{o})}[\log \pi_\theta(\mathbf{a}|\mathbf{o})]
\end{equation}
We again use the score function estimator for the gradient of the conditional entropy with respect to the strategy parameters 
\begin{equation}
    \label{eq:entropy_reg_gradient_estimate}
    \hat{\nabla_\theta \mathcal{H}(\theta)} = -\frac{1}{N}\sum_{j=1}^{N}W(\mathbf{a}^{(j)}|\mathbf{o}^{(j)})\nabla_\theta \log\pi_\theta(\mathbf{a}^{(j)}|\mathbf{o}^{(j)})
\end{equation}
where $W(\mathbf{a}^{(j)}|\mathbf{o}^{(j)}) = \log\pi_\theta(\mathbf{a}^{(j)}|\mathbf{o}^{(j)}) + 1$. The proof is similar to the one for \eqref{eq:non_local_gradient_estimate} and is provided in \ref{appendix:entropy_reinforce_derivation}. We then combine the expressions from \eqref{eq:non_local_gradient_estimate} and \eqref{eq:entropy_reg_gradient_estimate} which yields a simple surrogate objective of the form 
\begin{equation}
    \begin{aligned}
        \hat{J}(\theta) = \frac{1}{N}\sum_{j=1}^{N}&\text{sg}(V(\mathbf{a}^{(j)}|\mathbf{o}^{(j)}) -
        \alpha W(\mathbf{a}^{(j)}|\mathbf{o}^{(j)}))\log\pi_\theta(\mathbf{a}^{(j)}|\mathbf{o}^{(j)})
    \end{aligned}
\end{equation}
where sg() denotes the stopgradient operation.

\subsection{Derivations of Gradient Estimators}
\subsubsection{Win Probability}
\label{appendix:nonlocal_reinforce_proof}
Let the win probability for the nonlocal game, $J(\theta)$, be defined as in Equation \ref{eq:non_local_game_obj} and let $p_{\theta}(\mathbf{a}|\mathbf{o} = \mu(\mathbf{o})\pi_{\theta}(\mathbf{a}|\mathbf{o})$. Then
    \begin{equation*}
        \begin{aligned}
        \nabla_\theta J(\theta) &= \nabla_\theta \sum_\mathbf{o} \mu(\mathbf{o}) \sum_\mathbf{a} \pi_{\theta}(\mathbf{a}|\mathbf{o}) V(\mathbf{a}|\mathbf{o}) \\ 
        &= \sum_\mathbf{o} \mu(\mathbf{o}) \sum_\mathbf{a} \nabla_\theta\pi_{\theta}(\mathbf{a}|\mathbf{o}) V(\mathbf{a}|\mathbf{o}) \\ 
        % &= \sum_o \mu(o) \sum_a \pi_{\theta}(a|o) V(a|o) \nabla_\theta \log \pi_{\theta}(a|o) \\
        &= \sum_{\mathbf{o}, \mathbf{a}} \mu(\mathbf{o})\pi_{\theta}(\mathbf{a}|\mathbf{o}) V(\mathbf{a}|\mathbf{o}) \nabla_\theta \log \pi_{\theta}(\mathbf{a}|\mathbf{o}) \\
        &= \mathbb{E}_{\mathbf{o}, \mathbf{a} \sim p_{\theta}}[V(\mathbf{a}|\mathbf{o})\nabla_\theta \log\pi_{\theta}(\mathbf{a}|\mathbf{o})]
    \end{aligned}
\end{equation*}
since $\nabla_\theta \pi_{\theta}(\mathbf{a}|\mathbf{o}) = \pi_{\theta}(\mathbf{a}|\mathbf{o})\nabla_\theta \log\pi_{\theta}(\mathbf{a}|\mathbf{o})$. We then approximate this expectation using the Monte Carlo method, as we can collect samples from $p_{\theta}$ using ancestral sampling (since playing a round of the game $G$ draws a sample $\mathbf{o}$ from $\mu$ and then the players sample an action given their strategy $\pi_\theta(\cdot|\mathbf{o})$).

\subsubsection{Conditional Entropy}
\label{appendix:entropy_reinforce_derivation}
As mentioned in \ref{appendix:entropy_regularization}, we add a conditional entropy term to the objective function. 
We now provide the derivation of the estimator \eqref{eq:entropy_reg_gradient_estimate} 
of the gradient of the conditional entropy.

Let $W(\mathbf{a}|\mathbf{o}) = \log\pi_\theta(\mathbf{a}|\mathbf{o}) + 1$ as before. We then have that
    \begin{equation*}
    \begin{aligned}
        \nabla_\theta \mathcal{H}(\theta) &=  -\nabla_\theta \sum_\mathbf{o} \mu(\mathbf{o}) \sum_\mathbf{a} \pi_{\theta}(\mathbf{a}|\mathbf{o})\log\pi_\theta(\mathbf{a}|\mathbf{o})\\
        &= -\sum_\mathbf{o} \mu(\mathbf{o})\sum_\mathbf{a} W(\mathbf{a}|\mathbf{o})\nabla_\theta \pi_{\theta}(\mathbf{a}|\mathbf{o})\\
        &= -\sum_{\mathbf{o}, \mathbf{a}} \mu(\mathbf{o}) \pi_\theta(\mathbf{a}|\mathbf{o})W(\mathbf{a}|\mathbf{o})\nabla_\theta \log\pi_\theta(\mathbf{a}|\mathbf{o})\\
        &= -\mathbb{E}_{\mathbf{o}, \mathbf{a} \sim p_\theta}[W(\mathbf{a}|\mathbf{o})\nabla_\theta \log\pi_\theta(\mathbf{a}|\mathbf{o})]
    \end{aligned}
\end{equation*}
Thus, we can use the Monte Carlo method again to approximate this expectation using samples from $p_\theta(\mathbf{a}|\mathbf{o})$. 

\subsection{Game Details}
\subsubsection{CHSH}
\label{appendix:chsh}
\begin{table}[t]
\caption{
    CHSH Game Predicate Function, $a\oplus b=x\wedge y$.
    The two players win if they return the equal bits, except in the case where both players receive 1, in which case the win condition is for the players to return opposite bits.
}
\begin{center}
\renewcommand{\arraystretch}{1.2}  % Increases row height by 20%
\begin{tabular}{cc|cccc}
\multicolumn{2}{c}{} & \multicolumn{4}{c}{\textbf{$ab$}} \\
& \multicolumn{1}{c}{} & \textbf{00} & \textbf{01} & \textbf{10} & \textbf{11} \\
\cline{3-6}
& \textbf{00} & 1 & 0 & 0 & 1 \\
\multirow{2}{*}{\rotatebox{90}{\textbf{$xy$}}} & \textbf{01} & 1 & 0 & 0 & 1 \\
& \textbf{10} & 1 & 0 & 0 & 1 \\
& \textbf{11} & 0 & 1 & 1 & 0 \\
\end{tabular}
\end{center}
\label{tab:chsh_payoff}
\end{table}

The CHSH game is perhaps the canonical example of a nonlocal game in which there exists a strategy that leverages quantum entanglement to yield higher win probabilities than is classically possible. 
The CHSH game involves two players who each receive a random ``question" bit from the referee and then, without communication, both players send ``answer" bits to the referee in response. 
The players either jointly win or jointly lose depending on the combination of the question and answer bits. 
Let $x$ and $y$ denote the question bits that the two players receive, and $a$ and $b$ the answer bits that they return. 
The players win if the four bits satisfy the condition $a \oplus b = x \land y$ (where $\land$ denotes a logical $\operatorname{AND}$ operation and $\oplus$ denotes a logical $\operatorname{XOR}$ operation) and lose otherwise. 
The corresponding predicate function $V(a_i|o_i)$ for the CHSH game is provided in Table \ref{tab:chsh_payoff} for clarity. 

Players playing optimally without quantum entanglement will win 3/4 of the time.
(One possible strategy achieving this is for both players to always return the bit 0.)
On the other hand, players with shared quantum entanglement have a strategy whereby they can win with probability $\cos^2(\pi/8) \approx 0.8536$ \cite{1313847}.

For completeness, we present a shared state $\rho$ and POVM measurements acting on $\rho$ that achieve this optimal win probability.
Call the two players $A$ and $B$.
Let the shared state $\rho$ be the Bell state given in \eqref{eq:bell_state}.
When player $A$ receives bit 0 from the referee, she performs the POVM measurement
\begin{equation}
\{M(0|0),M(1|0)\}
=
\left\{
    \begin{bmatrix}
        1&0\\0&0
    \end{bmatrix},
    \begin{bmatrix}
        0&0\\0&1
    \end{bmatrix}
\right\}
\end{equation}
and bases her returned bit on the outcome.
When player $A$ receives bit 1, she performs the POVM measurement
\begin{equation}
\{M(0|1),M(1|1)\}
=
\left\{
    \begin{bmatrix}
        \frac{1}{2}&\frac{1}{2}\\\frac{1}{2}&\frac{1}{2}
    \end{bmatrix},
    \begin{bmatrix}
        \frac{1}{2}&-\frac{1}{2}\\-\frac{1}{2}&\frac{1}{2}
    \end{bmatrix}
\right\}
\end{equation}
and gives a returned bit according to the outcome.
Player $B$ performs a similar strategy but with POVM
\begin{equation}
\{N(0|0),N(1|0)\}
=
\left\{
    \begin{bmatrix}
        \cos^2\theta&\sin2\theta\\
        \sin2\theta&\sin^2\theta
    \end{bmatrix},
    \begin{bmatrix}
        \sin^2\theta&-\sin2\theta\\
        -\sin2\theta&\cos^2\theta
    \end{bmatrix}
\right\}
\end{equation}
when given bit 0 and POVM
\begin{equation}
\{N(0|1),N(1|1)\}
=
\left\{
    \begin{bmatrix}
        \cos^2\theta&-\sin2\theta\\
        -\sin2\theta&\sin^2\theta
    \end{bmatrix},
    \begin{bmatrix}
        \sin^2\theta&\sin2\theta\\
        \sin2\theta&\cos^2\theta
    \end{bmatrix}
\right\}
\end{equation}
when given bit 1, where $\theta=\pi/8$.
Using the Born rule a straightforward calculation shows that, for any combination of bits $o_A$ and $o_B$, the probability of either of the two corresponding win conditions is always $\frac{1}{2}\cos^2(\frac{\pi}{8})$, giving an overall win probability of $\cos^2(\frac{\pi}{8})$.

In other words, this choice of $\rho$ and POVM measurements allows players $A$ and $B$ to sample from a probability distribution $p(a_A,a_B|o_A,o_B)=\tr(\rho\,M(a_A|o_A)\!\otimes\!N(a_B|o_B))$ that cannot be written in the form $\sum_\lambda p(\lambda)p(a_A|o_A,\lambda)p(a_B|o_B,\lambda)$, the form which all decentralized, communication-free sampling must take when the players do not share entanglement.

\subsubsection{GHZ}
The GHZ game, as formulated by \citet{vaidman1999variations}, is an often cited example of quantum pseudo-telepathy, where players can derive a quantum strategy that wins the game with certainty (probability 1.0), whereas classically the win probability of optimal strategy is 0.75. However, unlike the CHSH game, the GHZ game comprises three players and a referee. Denote the bit received by the third player as $z$ and their response bit as $c$. The players win the game if $x \lor y \lor z = \text{mod}_2(a + b + c)$ (where $\lor$ denotes the logical $\operatorname{OR}$ operation and $\text{mod}_2(a + b + c)$ returns the sum of the response bits modulo 2) and lose otherwise. In addition, the referee only sends one of \{000, 110, 101, 011\} as the combination of questions to the players. We include results for the GHZ game to demonstrate that our framework is not limited to two-player games and that quantum advantage can still be obtained in this setting.

\subsubsection{Rendezvous}
The rendezvous game is another nonlocal game for which classical and quantum strategies have been analyzed. The game, as defined by \citet{PhysRevA.109.042201}, is played on a graph, where $r \geq 2$ players are placed on a starting vertex (chosen randomly and uniformly by the referee) and need to move along the edges of the graph $h \geq 1$ times, sending their moves to the referee. The players are unaware of the location of the other players on the graph at all times and are not allowed to communicate with one another. The players win the game if after the $h$ moves per player, they all end up at the same vertex. In the following experiments, we set $h = 1$ and $r = 2$ to compare to the results from \cite{PhysRevA.109.042201}.

The first graph analyzed, which we refer to as ``Tetra", can be seen in Figure \ref{fig:rendezvous1_graph}, which was also studied in \cite{PhysRevA.109.042201}, where they state the win probability under the optimal classical strategy is 0.6250 and also provide an upper bound for the win probability for the optimal quantum strategy as 0.64506 using the NPA method from \cite{PhysRevLett.98.010401}. 

The rendezvous game was also analyzed for another graph, ``Cube", depicted in Figure \ref{fig:rendezvous2_graph}, where the win probability under the optimal classical strategy is 0.3125 and the NPA upper bound for the win probability for the optimal quantum strategy is 0.32253. 

\begin{figure}[!ht]
    \centering
    \begin{subfigure}[t]{0.48\textwidth}
        \centering
        \includegraphics[width=\linewidth]{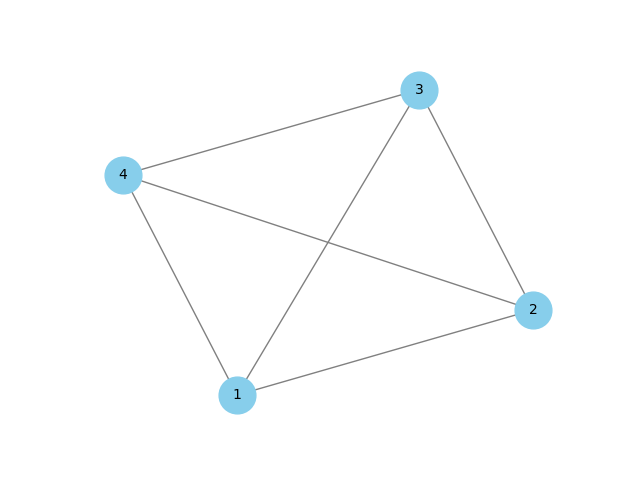}
        \caption{Experiment 1 (Tetra)}
        \label{fig:rendezvous1_graph}
    \end{subfigure}
    \hfill
    \begin{subfigure}[t]{0.48\textwidth}
        \centering
        \includegraphics[width=\linewidth]{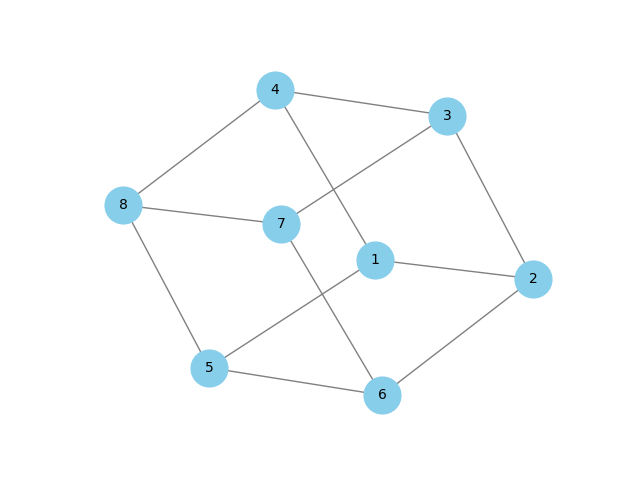}
        \caption{Experiment 2 (Cube)}
        \label{fig:rendezvous2_graph}
    \end{subfigure}
    \caption{Graphs used for the Rendezvous experiments.}
    \label{fig:rendezvous_graphs}
\end{figure}

\subsection{Experiment Details}
\label{appendix:experiment_details}
In order to learn strategies that leverage shared quantum entanglement, we used strategies of the form \eqref{eq:quantum-policy}, where we parametrized the POVM logits (the complex valued matrices used as inputs to QuantumSoftmax) using a neural network and subsequently used QuantumSoftmax, presented in Algorithm \ref{alg:quantum-softmax}, to obtain valid POVMs. In addition, we used a neural network to parameterize a matrix $B \in \mathbb{C}^{d\times d}$ and then obtain the density matrix $\rho = \frac{1}{c}B^H B$ where $c = \mathrm{tr}(B^H B)$, ensuring that $\rho$ is PSD and $\mathrm{tr}(\rho) = 1$. Our models were implemented in Jax using the Adam optimizer \cite{kingma2014adam} with gradients computed according to \eqref{eq:non_local_gradient_estimate} and \eqref{eq:entropy_reg_gradient_estimate}. For each game, we used 30 random initializations of our model, with and without entropy regularization (initializations kept consistent for means of comparison), resulting in training 60 models per game. Relevant hyperparameters were kept the same for all games, with a batch size of 512, a learning rate of 3e-2, the entropy regularization coefficient set to 0.2 (0.0 when entropy regularization is turned off), and the number of optimization steps equal to 5000.

After training was completed, we selected the best strategy (in terms of the win probability of the nonlocal game) over the 5000 training steps and computed its quantum advantage, defined as the difference between the win probability of the learned strategy and the win probability of the classically optimal strategy. We then normalize this learned advantage by the maximum possible advantage that can be obtained (using the theoretical upper bounds on the win probability for a quantum strategy for CHSH and GHZ and the results from the NPA method for the rendezvous games). In Table \ref{tab:nonlocal_results}, we report the quantum advantage (as a percentage of the maximum advantage that can be obtained) for the worst of the 30 training runs with and without entropy regularization for each of the nonlocal games. We observe that without entropy regularization, at least one of the 30 runs does not outperform the classically optimal solution (and thus has 0 quantum advantage) for each of the nonlocal games. However, when entropy regularization is added, all training runs learn strategies with quantum advantage, successfully leveraging entanglement to obtain winning probabilities higher than the best classical strategy. 

\begin{table}[ht]
    \caption{Learned Quantum Advantage as a Percentage of the Maximal Advantage for the Worst Training Run}
    \label{tab:nonlocal_results}
    \begin{center}
    \begin{tabular}{|c|c|c|c|c|}
    \hline
    Entropy? & CHSH & GHZ & Tetra & Cube \\
    \hline
    \newcrossmark & 0.00\% & 0.00\% & 0.00\% & 0.00\% \\
    \hline
    \newcheckmark & 99.90\% & 98.60\% & 84.25\% & 40.88\% \\
    \hline
    \end{tabular}
    \end{center}
\end{table}

\section{Derivation of Modified MAPPO Surrogate Objective Function}
\label{appendix:surrogate_objective_function}
The surrogate objective function of PPO, before clipping, is
\begin{equation}
L(\theta) = \mathbb{E}_\tau\left[
    \sum_t
    \frac{
        \pi_{\theta}(a_t|h_t)
    }{
        \pi_{\theta_0}(a_t|h_t)
    }A_t
\right],
\end{equation}
where $A_t$ is the advantage and the expectation is over trajectories $\tau$ induced by the policy with parameters $\theta_0$.
This is chosen so that its gradient evaluated at $\theta_0$ is the policy gradient estimate
\begin{equation}\label{eq:ppo_policy_gradient}
\nabla_\theta J(\theta)\big\rvert_{\theta_0}
=
\mathbb{E}_{\tau}\left[
    \sum_t
    \nabla_{\theta}\log\pi_{\theta}(a_t|h_t)
    A_t
\right]\Bigg\rvert_{\theta_0}.
\end{equation}
Replacing the policy in \eqref{eq:ppo_policy_gradient} here with the form $\pi(\mathbf{a}_t|\mathbf{h}_t)=\sum_{\mathbf{x}_t}q(\mathbf{x}_t|\mathbf{h_t})\prod_i\pi_i(a_{it}|x_{it},h_{it})$ gives the policy gradient
\begin{equation}
\begin{aligned}
\nabla_\theta J(\theta)
&=
\mathbb{E}_{\tau}\left[
    \sum_t
    \frac{
        \nabla_\theta q_\theta(\mathbf{x}_t|\mathbf{h}_t)
    }{
        q_\theta(\mathbf{x}_t|\mathbf{h}_t)
    }A_t
    % \right]\\
    \right]
    % &\quad+
    +
    \mathbb{E}_\tau
    \left[
    \sum_t
    \sum_i
    \frac{
        \nabla_{\theta}\pi_{i\theta}(a_{it}|x_{it},h_{it})
    }{
        \pi_{i\theta}(a_{it}|x_{it},h_{it})
    }A_t
\right],
\end{aligned}
\end{equation}
where the trajectories $\tau$ are now understood to include sampled values of $\mathbf{x}_t$.
This suggests a surrogate objective function of two terms, the familiar one for the actor models for the policies $\pi_{i}(a_{i}|x_i,h_{i})$,
\begin{equation}
    L_\text{actor}(\theta)
    =
    \mathbb{E}_\tau\left[\sum_t\sum_i
    \frac{
        \pi_{i\theta}(a_{it}|x_{it},h_{it})
    }{
        \pi_{i\theta_0}(a_{it}|x_{it},h_{it})
    }
    A_t
    \right],
\end{equation}
as well as a new term for the coordinator model for the distribution $q(\mathbf{x}|\mathbf{h})$,
\begin{equation}
    L_\text{coord}(\theta)
    =
    \mathbb{E}_\tau\left[\sum_t
    \frac{
        q_\theta(\mathbf{x}_t|\mathbf{h}_t)
    }{
        q_{\theta_0}(\mathbf{x}_t|\mathbf{h}_t)
    }
    A_t
    \right].
\end{equation}
Again, these expectations are over trajectories that include sampled values of $\mathbf{x}_t$ in addition to sampled states, observations, and actions.
These terms can each then be clipped as in PPO.

\section{Physical Implementation}
\label{appendix:physical_implementation}
Crucial to the ideas of this paper is the fact that it is possible for decentralized, non-communicating agents to sample from a policy of the form in \eqref{eq:quantum-policy}.
However, doing so will in general require that the agents each have a quantum system of sufficient dimension, that the joint state of their systems is in fact $\rho$, and that they can each physically implement their measurement $M_i(\,\cdot\,|h_i)$ on their system.
At present, this is a nontrivial technological feat, even for simple cases. 
In practice, the quantum systems may be small, say 2-dimensional, $\rho$ may be less than fully entangled, and only a subset of measurements might be implementable.
Given any of these restrictions the possible policies $\sum_\mathbf{x}\trace\left(\rho\ \bigotimes_i\! M_i(x_i|h_i)\right)\prod_i\pi_i(a_i|x_i,h_i)$ will not include the full set of entangled policies.
Whatever the limitations, however, the methods of this paper provide a path for learning strategies that incorporate the precise quantum resources that are available, by modeling $\rho$ and the POVMs appropriately.

When implementing with physical systems one is faced with the potential difficulty of calculating $q_\theta(\mathbf{x}|\mathbf{h})=\trace\left(\rho_\theta\ \bigotimes_i\! M_{i\theta}(x_i|h_i)\right)$, as well as its gradient with respect to $\theta$, for given $\mathbf{a}$ and $\mathbf{h}$. 
These are necessary for calculating the objective function gradient in the modified MAPPO algorithm we describe.
Measuring a quantum state leaves the state changed, so while values of $q_\theta(\mathbf{x}|\mathbf{h})$ could be estimated using the physical system by taking many samples, this would require re-preparing the state $\rho$ many times.
One solution to these problems is simulation, obtaining models of the physical state and possible measurements through tomography.
For the small quantum systems possible with current technology, such simulation is not prohibitive.
The simulation would be necessary only for training and not for deployment, where only sampling from $q_\theta(\mathbf{x}|\mathbf{h})$ is necessary.
This sampling of course must be done with the physical entangled quantum system in order to be done in a decentralized way without communication.

Creating an entangled state cannot, itself, be done in a decentralized way without communication, however.
If the agents cannot communicate when deployed, the states $\rho$ must be pre-shared.
In practice, each agent could have a quantum memory with many copies of their part of $\rho$ stored ahead of time.
At each time step a copy of $\rho$ would be consumed as the agent made a measurement that time step.

In cases where the communication constraint is simply due to latency \cite{ding2024coordinating, ding2025quantumnonlocalitylatencyconstraints}, i.e.\ where communication is possible, but slower than the timescale at which decisions must be made, one could share entangled states in real time, having new entangled states arrive continuously.
Practically this could be implemented with entangled photons sent through optical fiber.

\section{The Quantum Coordinator Ansatz}
%\section{Quantum Coordinators}
\label{appendix:quantum_coordinators}
Any policy obtainable with a shared quantum state $\rho$ can be obtained from the same state through the quantum coordinator ansatz described in \eqref{eq:coordinated-policy-parameterized}.
\begin{proposition}\label{advice_sufficient}
    Any joint policy of the form $\pi(a_1\cdots a_n|h_1\cdots h_n)=\trace\left(\rho\,\bigotimes_i^nM_i(a_i|h_i)\right)$ for POVMs $M_i(\cdot|h_i)$ can also be written in the form $\sum_x\trace(\rho\,\bigotimes_i^n\tilde{M_i}(x_i|h_i))\prod_i\pi_i(a_i|x_i,h_i)$ for some POVMs $\tilde{M}_i(\cdot|h_i)$ and local policies $\pi_i(a_i|x_i,h_i)$.
    The reverse is also true.
\end{proposition}
\begin{proof}
    Suppose $\pi(a_1\cdots a_n|h_1\cdots h_n)$ has the form $\trace\left(\rho\,\bigotimes_i^nM_i(a_i|h_i)\right)$ described.
    Then simply choose $\tilde{M_i}(\cdot|h_i)=M_i(\cdot|h_i)$ for all $h_i$ and let $\pi_i(a_i|x_i,h_i)=\delta_{a_ix_i}$.
    Now suppose instead that $\pi(a_1\cdots a_n|h_1\cdots h_n)$ has the form $\sum_x\trace(\rho\,\bigotimes_i^n\tilde{M}_i(x_i|h_i))\prod_i\pi_i(a_i|x_i,h_i)$.
    Undistributing the sum over $x=x_1\cdots x_n$ gives
    \begin{equation*}
    \begin{aligned}
        \pi(a_1\cdots a_n|h_1\cdots h_n)
        &=
        \sum_x\trace\left(\rho\,\bigotimes_i^n\tilde{M}_i(x_i|h_i)\right)\prod_i\pi_i(a_i|x_i,h_i)\\
        &=
        \trace\left(
        \rho\,\bigotimes_i^n\!\left(\sum_{x_i}\pi_i(a_i|x_i,h_i)\tilde{M}_i(x_i|h_i)\right)
        \right).
    \end{aligned}
    \end{equation*}
    Let $M_i(a_i|h_i)=\sum_{x_i}\pi_i(a_i|x_i,h_i)\tilde{M}_i(x_i|h_i)$.
    The matrices $M_i(a_i|h_i)$ so chosen form a POVM: each $M_i(a_i|h_i)$ is PSD as it is a positive combination of PSD matrices $\tilde{M}_i(x_i|h_i)$, and they sum to the identity because
    \begin{equation*}
    \begin{aligned}
        \sum_aM_i(a|h_i)
        &=
        \sum_a\sum_{x_i}\pi_i(a|x_i,h_i)\tilde{M}_i(x_i|h_i)\\
        &=
        \sum_{x_i}\left(\sum_a\pi_i(a|x_i,h_i)\right)\tilde{M}_i(x_i|h_i)\\
        &=
        \sum_{x_i}\tilde{M}_i(x_i|h_i)\\
        &=
        I.
    \end{aligned}
    \end{equation*}
\end{proof}
Thus the quantum coordinator ansatz, at least without further constraints on the allowed POVMs $\tilde{M}_i(x_i|h_i)$ is not constraining.
This is only true if we allow the POVMs $\tilde{M}(\cdot|h_i)$ to have sufficiently many outcomes (i.e.\ if we allow $x$ to take sufficiently many values).
In fact, it no longer holds if the number of outcomes of the $\tilde{M}(\cdot|h_i)$ are at all less than the number of outcomes of $M(\cdot|h_i)$.
This is implied by the following result:
\begin{proposition}\label{advice_insufficient}
    For any positive integers $k$ and $d$, there exists a $k$-outcome, $d$-dimensional POVM $M(\cdot)$ that cannot be written in the form $M(a)=\sum_{b}P_{ab}\tilde{M}(b)$ where $P$ is a stochastic matrix and $\tilde{M}(\cdot)$ is a POVM with fewer than $k$ outcomes.
\end{proposition}
\begin{proof}
    We will consider the cases $k\leq d$ and $k>d$ separately.
    \begin{description}
    \item[Case $k\leq d$.]
    Choose $M(a)$ to be $k$ commuting projectors in $\mathbb{C}^{d\times d}$ whose sum is $I_{d\times d}$.
    % Note that this implies $M(a)M(a')=\mathbf{0}$ for $a\neq a'$.
    Suppose that $M(a)=\sum_bP_{ab}\tilde{M}(b)$ where $\tilde{M}(\cdot)$ is a POVM and $P$ is a stochastic matrix.
    For every $M(a)$ choose a unit vector $v_a$ entirely in the subspace defined by $M(a)$.
    % Then $v_a^\dagger M(a')v_a=\delta_{aa'}$.
    For every $a$ there must be some $\t{M}(b_a)$ with positive support in the sum $\sum_bP_{ab}\t{M}(b)$ and where $v_a^\dagger\t{M}(b_a)v_a$ is not zero (else $v_a^\dagger M(a')v_a$ would be zero).
    For any other $a'$ distinct from $a$, the support of $\t{M}(b_a)$ in the sum $M(a')=\sum_{b}P_{a'b}\t{M}(b)$ must be zero, otherwise $v_a^\dagger M(a')v_a=\sum_bP_{a'b}v_a^\dagger\t{M}(b)v_a$ would be nonzero as it would include the positive term $P_{a'b_a}v_a^\dagger\t{M}(b_a)v_a$ together in a sum with all non-negative terms.
    Thus each $a$ corresponds to some $b_a$, and distinct $a$ correspond to distinct $b_a$.
    This implies there are at least as many values of $b$ as $a$.
    \item[Case $k>d$.]
    Consider the matrices 
    \begin{equation*}
    \begin{aligned}
    M(a)
    &=
    \frac{1}{k}
    \begin{bmatrix}
    1 & \omega^{-a} & \omega^{-2a} & \cdots & \omega^{-(d-1)a}
    \end{bmatrix}^\top
    \begin{bmatrix}
    1 & \omega^a & \omega^{2a} & \cdots & \omega^{(d-1)a}
    \end{bmatrix}
    \\&=
    \frac{1}{k}
    \begin{bmatrix}
    1 & \omega^{a} & \cdots & \omega^{(d-1)a} \\
    \omega^{-a} & 1 & \cdots & \omega^{(d-2)a} \\
    \vdots & \vdots & \ddots & \vdots \\
    \omega^{-(d-1)a} & \omega^{-(d-2)a} & \cdots & 1
    \end{bmatrix}
    \end{aligned}
    \end{equation*}
    where $a$ runs from $1$ to $k$ and where $\omega=e^{2\pi i/k}$.
    These matrices form a POVM, as they are manifestly positive semidefinite and sum up to the identity by construction.
    Suppose that $M(a)=\sum_bP_{ab}\tilde{M}(b)$ as described.
    For any $a$, the matrix $M(a)$ is thus a positive combination of some of the matrices $\tilde{M}(b)$.
    Because the matrix $M(a)$ has rank one, any of the $\tilde{M}(b)$ with positive support in the combination must be proportional to $M(a)$.
    (All of the $\tilde{M}(b)$ with positive support must be proportional to each other, lest the sum have rank greater than one.
    Then, for their sum to equal $M(a)$, they must all be proportional to $M(a)$.)
    Thus at least one of the matrices $\tilde{M}(b)$ must be proportional to $M(a)$.
    This is true for any $a$.
    None of the matrices $M(a)$ are proportional to each other, so there must be at least as many distinct matrices $\tilde{M}(b)$ as matrices $M(a)$.
    \end{description}
\end{proof}

\section{The Multi-Router Queueing Problem and Experiment}
\label{appendix:server-router}
In this appendix we first describe our formulation of the multi-router queueing problem introduced in \cite{da2025entanglement, dasilva2026entanglement} as an MDP.
We then describe details of our application of a modified MAPPO method to learn strategies that use quantum entanglement.

\subsection{Multi-Router Queueing as an MDP}
Recall the setup of two routers that must each decide where to send incoming customer requests.
The problem statement of \cite{dasilva2026entanglement} specifies that customer requests arrive in pairs (one to each router) simultaneously.
The routers make a routing decision upon arrival of the new requests, so each new arrival is a new time step of the MDP.

\subsubsection{States}\label{states}
Each server queue has its own queue state.
A queue is either empty (so that the server is working on the baseline task) or full (so that the server is working on a customer request).
If empty, the queue state includes a continuous variable specifying how long the server has been working on the baseline task uninterrupted.
If full, the queue state includes a continuous variable specifying how much time the current queue of customer requests will take to complete.
All this can be combined into a single continuous variable $q$.
When $q$ is positive it signifies a full queue, and the magnitude of $q$ represents the total size of the customer tasks left in the queue.
When $q$ is negative it signifies an empty queue, and the magnitude represents the unbroken amount of time the server has been currently working on the baseline task.
The MDP state is the joint state $(q_1, q_2)$ of the two queues when the new customer request pair arrives.

\subsubsection{Actions}\label{actions}
During a single time step the routers receive customer requests of sizes $x_1$, $x_2$ respectively.
They then take actions $a_1, a_2\in \{0,1\}$ representing the server(s) to which they route the requests, where $0$ represents the the left server (say) and $1$ represents the right server.
The action determines how much total load is sent to a given server.
The new load on the left server is $(1-a_1)x_1+(1-a_2)x_2$ and the new load on the right server is $a_1x_1+a_2x_2$.
As stipulated by \citet{dasilva2026entanglement}, the agents observe nothing other than their current, local values of $x$ when choosing their actions, and they do not retain histories of any past observations.

\subsubsection{State transitions}\label{state_transitions}
The state $q$ of a queue transitions to a new state $q\rightarrow q'$ every time step.
Let $\Delta q$ be the additional customer request load sent to the server and let $\Delta t$ be the time before the next pair of customer requests.

Suppose we start with negative $q$, so the server is currently working on the baseline task and has been for time $|q|$.
If no new requests are sent to the server ($\Delta q=0$), we simply subtract $\Delta t$ from $q$ to the new state, representing the fact that the server, at the next time step, will have been working uninterrupted on the baseline task for an additional time $\delta t$.
So in this case $q\rightarrow q'=q-\Delta t$
On the other hand, if $\Delta q$ is non-zero, the baseline task is interrupted and the new state is $q'=\Delta q - \Delta t$, either representing the new queue size if $\Delta q>\Delta t$ or the uninterrupted time working on the baseline task if $\Delta q<\Delta t$.

Now suppose we start with $q$ positive.
Then after the adding the load $\Delta q$ and waiting the time $\Delta t$ until the next time step, the new state is $q'=q+\Delta q-\Delta t$, which may be either positive or negative.

All this can be captured by the following rule
\begin{equation}\label{queue_transition}
\begin{aligned}
q\rightarrow
\Delta q-\Delta t
+\begin{cases}
0\text{\quad if $q<0$ and $\Delta q>0$}\\
q\text{\quad otherwise}
\end{cases}
\end{aligned}
\end{equation}
In other words, the new queue state is $q'=q+\Delta q-\Delta t$ except in the case when an incoming request interrupts an ongoing baseline task and starts a new queue.
In that case the new queue state is $q'=\Delta q-\Delta t$ instead.

The amount of new load $\Delta q$ added to a queue depends on the actions of the routers.
Specifically we have $\Delta q_1=(1-a_1)x_1+(1-a_2)x_2$ and $\Delta q_2=a_1x_1+a_2x_2$.
So in all the joint state transition rule, given joint action $(a_1,a_2)$, is
\begin{equation}\label{eq:full_transition_rule}
\begin{aligned}
    q_1&\rightarrow
    (1-a_1)x_1+(1-a_2)x_2-\Delta t
    +
    \left\{\begin{rcases}
        0\\
        q_1
    \end{rcases}\right.\\
    q_2&\rightarrow
    a_1x_1+a_2x_2-\Delta t
    +
    \left\{\begin{rcases}
        0\\
        q_2
    \end{rcases}\right.
\end{aligned}
\end{equation}
with the upper case chosen for router $i$ if $q_i<0$ and $\Delta q_i>0$.
Note everything here is either the action variables $a_i$ or the random variables $x_i$ and $\Delta t$.
Following \cite{dasilva2026entanglement}, the random variables $x_i$ are chosen independently from an exponential distribution.
Similarly, the random variable $\Delta t$ is chosen from a different exponential distribution.
In our experiments we use the choices $x_i\sim \mu e^{-\mu x_i}$ where $\mu=1$ and $\Delta t\sim \lambda e^{-\lambda \Delta t}$ where $\lambda=0.8$.

\subsubsection{Rewards}
The reward for each time step is the sum of two independent rewards, one for each queue.
The servers are rewarded for throughput on the baseline task.
To model a case where the efficiency at the baseline task increases with uninterrupted work, the servers should be rewarded $T(t)$ for each uninterrupted interval, where $T(t)$ is a superlinear function of $t$.
In our experiments we choose the function $T(t)=t^2$.
We must formulate the MDP rewards for each time step such that the total over an unbroken interval of total time $t$ equals $T(t)$.
An interval of work on the baseline task might be spread over several consecutive MDP time steps, however, so there is some question about how and when to give out the reward for an interval.
For example, we could wait till the interval has ended to give the reward $T(t)$, giving out a reward of $0$ for intervening time steps.
Waiting possibly several time steps before giving out the reward could make training more difficult, however.
Instead, one could divvy up the reward so that the reward so far for an ongoing interval of current length $t$ is always $T(t)$.
For example, suppose that after a time step of time $\Delta t$ the server has been working on the task for a total time $t>\Delta t$.
Then for this particular time step give out the reward $T(t) - T(t-\Delta t)$.
With this reward convention in mind let's go through all the cases.

The reward for a queue will depend on the old state of the queue $q$ and the values $\Delta q$ and $\Delta t$ for that time step.
If $q$ was positive then the uninterrupted time reward is zero if $q'$ remains positive and $T(-q')$ if it is negative, as in this case the server has started the baseline task again and worked on it for $-q'$ time.
These can be summarized as $T([-q']_+)$ where $[x]_+:=\max\{0,x\}$.

If $q$ was negative then there are three cases to consider.
If $q'$ is positive then the server did not do any of the baseline task this time step, so the reward is $0$.
If $q-\Delta t<q'<0$ then the baseline task was interrupted and then started again after handling a customer request, so the reward is $-q'$.
Both of these cases are described by $T([-q']_+)$.
If $q'=q-\Delta t$, however, which happens when only $\Delta q=0$, the baseline task has continued uninterrupted from the previous time step.
The reward should then be $T(-q')-T(-q)$.

All the above cases can be summarized as
\begin{equation}\label{time step_rewards}
\begin{cases}
    T(-q')-T(-q)\text{\quad if $q<0$ and $\Delta q=0$}\\
    T([-q']_+)\text{\quad otherwise}
\end{cases}
\end{equation}
In other words, the reward is always $T$ of the positive part of $-q'$ except in the case where the baseline task is never interrupted, in which case the reward is $T(-q')-T(-q)$.
Handling the rewards in the above way should result in a separate reward of $T(t)$ for every uninterrupted interval of time $t$ encountered in a trajectory.

The above is the reward for a single server.
The total reward is
\begin{equation}\label{total_reward}
\begin{aligned}
R(q_1,q_2,a_1,a_2,q'_1,q'_2)
&=
\left\{
\begin{rcases}
    T(-q_1')-T(-q_1)\\
    T([-q_1']_+)
\end{rcases}
\right.
\\&+
\left\{
\begin{rcases}
    T(-q_2')-T(-q_2)\\
    T([-q_2']_+)
\end{rcases}
\right.
\end{aligned}
\end{equation}
where the top values are used respectively whenever $q_i<0$ and $\Delta q_i=0$.

\subsubsection{Note on observations}
Above we specified that the agents only observe their local values of $x$.
This is not the only possibility.
Though it is the assumption made by \cite{dasilva2026entanglement}, and the assumption we make in our MDP formulation used in our experiment, our method is flexible enough to allow for alternatives.
It is an interesting follow-up question to determine whether there is still quantum advantage under the various possibilities.

In addition to their local job sizes $x$, we could imagine the agents seeing the size of one or both queues or seeing $\Delta t$.
This last possibility is easily motivated, as in any real implementation of this set up the servers would indeed be able to measure the time $\Delta t$  between arrivals.
Knowing the inter-arrival time $\Delta t$ in addition to $x$ could, in the absence of direct observations of $q$, gives the agents partial information about how full the queues are.

\subsubsection{Quality of service constraint}
The authors \citet{dasilva2026entanglement} only consider solutions that satisfy the quality of service constraint that the long-term average customer wait time be less than a given constant $W_\ell$.
They define the wait time as the time until a server begins work on the customer's task.
A question arises as to how to reckon wait times when both customer requests are sent to the same server.
In this case the requests are randomly ordered and the second request will have its wait time increased by the size of the request that goes before it.

With this constraint on the customer wait time, we can consider the problem an example of a constrained MDP (CMDP).
If we consider a trajectory with $T$ steps, we allow the agents a ``budget" of $TW_\ell$ that constrains the total customer wait time $\sum_{t=1}^Tw_t$ where $w_t$ is the total time customers wait in time step t.

Such constrained MDPs can be handled by including a cost critic model in addition to the value critic and adding a cost penalty term to the objective function.

\subsubsection{Symmetry Constraint}
Finally, a crucial aspect of the problem as understood by \cite{dasilva2026entanglement} is an additional restriction that the routers' joint policy be ``balanced", i.e.\ invariant under a relabeling of the servers.
Explicitly, the joint policy $\pi(a_1,a_1|x_1,x_2)$ must have the following symmetry:
\begin{equation}
    \pi(a_1,a_2|x_1,x_2)=\pi(1-a_1,1-a_2|x_1,x_2).
\end{equation}
This rules out some strategies that might, depending on chosen parameters, otherwise be optimal, including strategies where routers send their jobs preferentially to one server or the other.
To implement this constraint in practice, the actions $a_1$, $a_2$ of the two servers are simultaneously transformed $a_1\rightarrow 1-a_1$, $a_2\rightarrow 1-a_2$ with probability $\frac{1}{2}$ at each step.
Thus the deterministic transition rule \eqref{eq:full_transition_rule} is effectively replaced by the stochastic one
\begin{equation}\label{eq:symmetric_transition_rule}
\begin{tikzcd}
    & \left(
        (1-a_1)x_1+(1-a_2)x_2-\Delta t+\left\{\begin{rcases}0\\q_1\end{rcases}\right.,
        a_1x_1+a_2x_2-\Delta t+\left\{\begin{rcases}0\\q_2\end{rcases}\right.
    \right) \\
    (q_1,q_2) \arrow[ru, "\frac{1}{2}", end anchor=west] \arrow[rd, "\frac{1}{2}", end anchor=west] & \\ 
    & \left(
        a_1x_1+a_2x_2-\Delta t+\left\{\begin{rcases}0\\q_1\end{rcases}\right.,
        (1-a_1)x_1+(1-a_2)x_2-\Delta t+\left\{\begin{rcases}0\\q_2\end{rcases}\right.
    \right)
\end{tikzcd},
\end{equation}
where, of course, we take the upper entry in braces when $q_i<0$ and $\Delta q_i=0$, as before.

\subsection{Experiment Procedure and Results}
We implemented the above MDP as an RL environment using the JaxMARL library \cite{flair2024jaxmarl}.
The quantum coordinator model simulated quantum measurements of dimension 2 with 2 outcomes, parametrized using $\mathsf{QuantumSoftmax}$ layers as described in \ref{sec:povm_softmax}.
For the shared quantum state $\rho$ we used the $2\times2$-dimensional Bell state \eqref{eq:bell_state}.
Recall also that we are considering a version of the problem where the throughput in terms of the unbroken interval length $t$ is $T(t)=t^2$, the rate parameter for the interarrival times $\Delta t$ is $\lambda=0.8$, and the rate parameter for the service request sizes $x_i$ is $\mu=1.0$.

Given the simplicity of the environment, in particular the fact that the action set has size two, we made the actor models trivial, simply passing on the advice received from the coordinator model.
This reduced redundancy in the architecture, as the coordinator model itself gave measurement outcomes out of 2 possible outcomes and had sufficient parameters.
We wish to stress that in more complicated environments, for example those with many actors or large or continuous action spaces, this would not be an adequate choice.

The results presented in Fig.~\ref{fig:server_router_perf} were obtained by first running 8 training runs with wait time constraint set to 5.5, choosing the best performing solution, then performing a sequence of runs, each warm-started from the previous one, with the wait time constraint incremented.
For each run the best and final models were saved, evaluated and plotted based no the evaluation.
The theoretical optimal wait times (for both the shared randomness and full-communication cases) which we benchmark against were calculated following the queueing theory arguments of \cite{dasilva2026entanglement}

\section{Numerical Stability of Quantum Softmax and Automatic Differentiation}
\label{sec:differentiability-quantum-softmax}
In Algorithm~\ref{alg:quantum-softmax}, we note that $\mathsf{QuantumSoftmax}$ requires taking the explicit inverse square root of the positive-definite matrix, the matrix $S=\sum_j\mathrm{expm}(\tilde{Z}_j)$, where $\{\tilde{Z}_j\}_{j\in [n]}$ are the POVM ``logits."
Beyond the positive-definiteness induced by the matrix exponential of the Hermitian POVM logits, there are no constraints on the eigenvalues of $S$ and, in practice, the condition number of $S$ can become arbitrarily large, causing numerical difficulties in calculating the inverse square root.
This issue can be remedied, in practice, by adding a conditioning term to the loss which penalizes $S$ for being far from the identity: $L_\text{cond}=(S-I)^\mathsf{H}(S-I)$. Doing so does not reduce expressiveness of parameterized POVMs that make use of~$\mathsf{QuantumSoftmax}$, given that any POVM $\{P_j\}_j$ can be written in terms of POVM logits $Z_j=\log(P_j)$ whose corresponding matrix $S$ is $S=\sum_j\exp(Z_j)=\sum_jP_j=I$.

\section{Coordinator-Advice Policies}

\begin{figure}[!ht]
    \centering
    \includegraphics[width=0.5\linewidth]{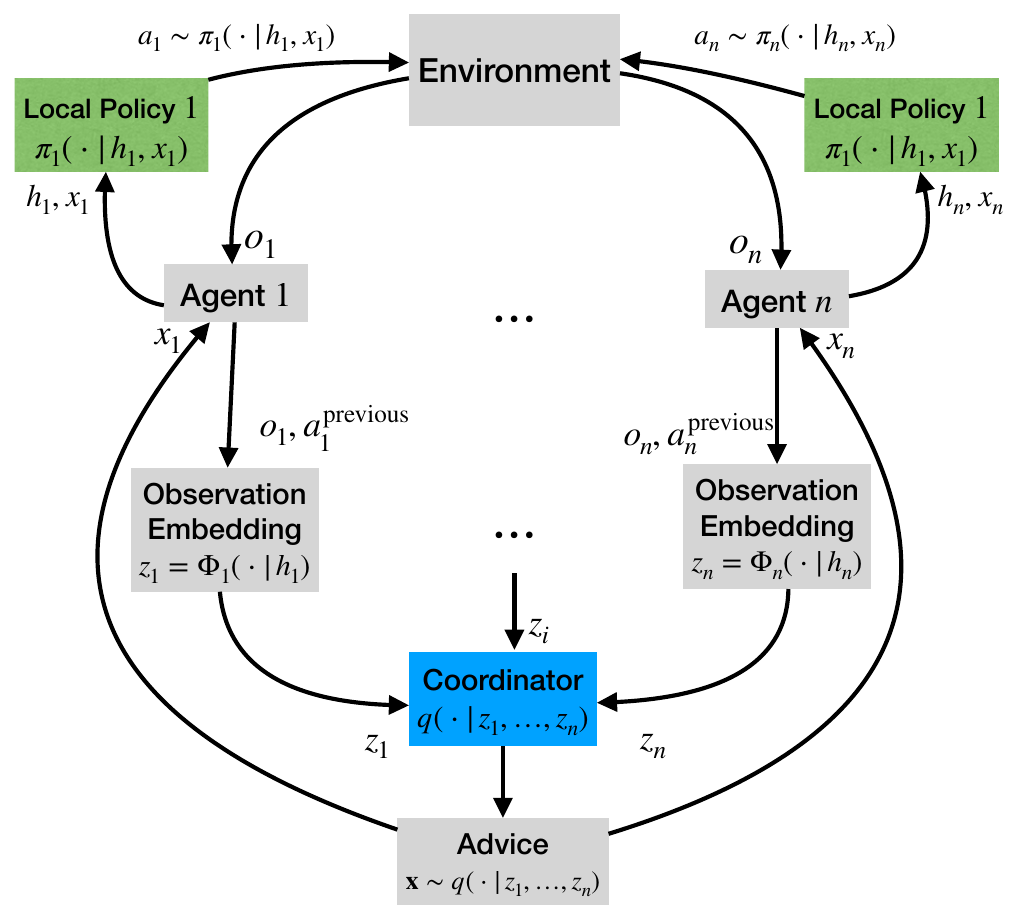}
    \caption{Coordination network.}
    \label{fig:coordinator}
\end{figure}

\end{document}